\documentclass[reqno,centertags,oneside]{amsart}

\usepackage[T1]{fontenc}
\usepackage{mathrsfs}
\usepackage{amsthm}
\usepackage{amsmath}
\usepackage{amssymb,amsfonts}
\usepackage{bbding}
\usepackage{bm}
\usepackage{enumerate}

\usepackage{graphics,graphicx,subfigure,color}
\usepackage{xcolor}

\usepackage{natbib}

\usepackage{txfonts}
\usepackage{textcomp}

\newcommand{\R}{\varmathbb{R}}

\newcommand{\RO}{\varmathbb{R}_{0}}
\newcommand{\C}{\varmathbb{C}}
\newcommand{\mrm}[1]{\mathrm{#1}}

\newcommand{\co}{\!\colon\thinspace}

\newcommand{\ess}{\sigma_{\mrm{ess}}}
\newcommand{\disc}{\sigma_{\mrm{disc}}}

\newcommand{\so}{\sigma_{\mrm{so}}}
\newcommand{\Wp}[2]{W^{#1}(#2)^{2}}
\newcommand{\WpO}[2]{W_{0}^{#1}(#2)^{2}}
\newcommand{\CO}[1]{C_{0}^{\infty}(#1)}
\newcommand{\Lp}[1]{L^{2}(#1)^{2}}

\newcommand{\ot}{\text{\large $\otimes$}}

\theoremstyle{plain}
{
\newtheorem{thm}{Theorem}[section]

\newtheorem{lem}[thm]{Lemma}
\newtheorem{prop}[thm]{Proposition}
}
\theoremstyle{definition}
{

}
\theoremstyle{remark} {
\newtheorem{rem}[thm]{Remark}

}
\renewcommand{\thesection}{{\bf\Roman{section}}}

\makeatletter
\renewcommand\section{\@startsection {section}{1}{\z@}%
{-3.5ex \@plus -1ex \@minus -.2ex}%
{2.3ex \@plus.2ex}%
{\Large\bfseries}}
\makeatother

\numberwithin{equation}{section}

\makeatletter
\renewcommand*{\@cite@ofmt}{\bfseries\hbox}
\makeatother

\boldmath
\selectfont

\begin{document}

\title[Bound states of the spin-orbit coupled ultracold atom]{
Bound states of the spin-orbit coupled ultracold atom in a one-dimensional short-range potential}

\author{Rytis Jur\v{s}\.{e}nas and Julius Ruseckas \\ \\ 
{\tiny Institute of Theoretical Physics and Astronomy \\ of Vilnius University, A. Go\v{s}tauto 12, LT-01108, Vilnius, Lithuania}}

\date{\today}

\begin{abstract}
We solve the bound state problem for the Hamiltonian with the spin-orbit and the Raman coupling included. The Hamiltonian is
perturbed by a one-dimensional short-range potential $V$ which describes the impurity scattering. In addition to the bound states
obtained by considering weak solutions through the Fourier transform or by solving the eigenvalue equation on a suitable domain directly, 
it is shown that ordinary point-interaction representations of $V$ lead to spin-orbit induced extra states.

\noindent{}\textbf{PACS(2010):} 03.65.Ge, 67.85.-d, 71.70.Ej
\end{abstract}


\maketitle

\section{Introduction}\label{sec:intro}

The study of ultracold atomic gases is one of the most actively developed areas of the physics of quantum many-body systems.
Initiated by the pioneering experiments with synthetic gauge fields in both Bose gases \citep{Lin11,Lin09} and Fermi gases \citep{Wang12},  
theoretical physicists took over the research for providing various schemes to synthesize certain extensions to Rashba--Dresselhaus
\citep{Rashba84,Dressel55} spin-orbit coupling for cold atoms \citep{Anderson12,Campbell11,Dalibard11,Juzeliun10}. As a result, one
derives a single-particle Hamiltonian of the form $-\Delta\ot  I+U$, where $\Delta$ is the Laplacian, $I$ is the 
identity operator in $\C^{2}$ (or $\R$), and $U$ is the atom-light coupling containing the spin-orbit interaction of the Rashba or 
Dresselhaus form and the Zeeman field. In a one-dimensional atomic center-of-mass motion, the simplified Hamiltonian of a particle
with mass $1/2$ (in $\hbar=c=1$ units) accedes to a formal differential expression in the configuration space $\R\ot \C^{2}$,

\begin{equation}
H=H_{0}+V(x)\ot  I,\quad H_{0}=-\Delta\ot  I+U,\quad
U=-i\eta\nabla\ot \sigma_{2}+(\Omega/2)\ot \sigma_{3}
\label{eq:Hamilt}
\end{equation}

\noindent{}($x\in\R;\Omega,\eta\geq0;\Delta=d^{2}/dx^{2};\nabla=d/dx$), where $\eta$ labels the spin-orbit-coupling strength, $\Omega$ 
results from the Zeeman field and is named by the Raman-coupling strength; $\sigma_{2}$, $\sigma_{3}$ are the Pauli matrices. In 
(\ref{eq:Hamilt}), $V$ obeys the meaning of a short-range disorder localized in the neighborhood of $x=0$.

It seems to be the first time when the spectral properties---and in particular bound states---of the Hamiltonian realized through 
(\ref{eq:Hamilt}) are considered in detail. For the most part, our attempt to provide the analysis of the spectral characteristics for 
the spin-orbit Hamiltonian is motivated by the work of \cite{Lin11}, where the authors examined the free Hamiltonian $H_{0}$ in 
$\R^{3}\ot \C^{2}$, with $\nabla$ in $x\in\R$, and calculated, particularly, the dispersion relation. In a recent report of 
\cite{Cheuk12} (see also \citep{Galitski13}) such a dispersion was shown to had been measured in $^{6}$Li. 

A straightforward calculation shows that the atom-light coupling $U$ is unitarily equivalent to 
$\eta D_{0}\equiv-i\eta\nabla\ot \sigma_{1}+(\Omega/2)\ot \sigma_{3}$ ($\sigma_{1}$,
$\sigma_{3}$ are the Pauli matrices), and the associated unitary transformation is $I\ot  e^{-i\theta\sigma_{3}}$, where 
$\theta\equiv3\pi/4\!\!\!\mod\pi$. The operator $D_{0}$, provided $\eta>0$, is nothing more than the free one-dimensional Dirac 
operator for the particle with spin one-half and mass $\Omega/(2\eta)$ (in $\hbar=c=1$ units); see \cite{Hughes97,Benvegnu94} for
the analysis of this operator. It turns out that $H$ in (\ref{eq:Hamilt}) can also be interpreted as being equivalent to the (operator) 
sum of the free Dirac operator plus a Schr\"{o}dinger operator $(-\Delta+V)\ot  I$. In particular, this means that, as the 
spin-orbit-coupling strength $\eta$ increases, $H/\eta$ approaches the one-dimensional massless Dirac operator in Weyl's form. For 
arbitrary $\eta>0$, however, one can show that $A_{0}/\eta$, with $A_{0}=U$ defined on a suitable domain (Sec.~\ref{sec:fermi}), is 
unitarily equivalent to $D_{0}+(1/\eta)V_{F}\ot  I$, the one-dimensional Dirac operator for the particle moving in Fermi 
pseudopotential (see (\ref{eq:VF})). This particular feature enables us to show that $H$ admits both continuous and discontinuous 
functions at a zero point. Throughout, by a (dis)continuous function $f$, one accounts for the property whether $f(0_{+})=f(0_{-})\equiv 
f(0)$ (continuity) or not (discontinuity), though $f$ is assumed to be defined on any subset of $\R\backslash\{0\}$.

Originally, one would naturally conjecture that the disorder $V$ is prescribed by a potential well with its minimum at $x=0$. A good 
survey of approximations by smooth potentials can be found, for example, in \citep{Hughes97}. Also, there are numerous works concerning 
the generalized point-interactions in one-dimension; see eg the papers of 
\cite{Ravelo12,Malamud12,Albeverio05,Coutinho04,Coutinho97,Seba86}, and also the 
citations therein. In the present paper, we assume that $V$ is approximated by the square-well of width $2\epsilon$ and depth 
$1/(2\epsilon)$ for some arbitrarily small $\epsilon>0$; the coupling strength of interaction is $\gamma\in\R$. Evidently, this is a 
familiar $\delta$-interaction. The one-dimensional Schr\"{o}dinger and Dirac operators with $\delta$-interaction are known to be 
well-defined via the boundary conditions for everywhere continuous functions. In our case we have a mixture, to some extent, of
Schr\"{o}dinger-like and Dirac-like operators. In Sec.~\ref{sec:extra} we argue that in such a case there is a possibility
that discontinuous eigenfunctions would appear.

To avoid the difficulties concerning the uniqueness of self-adjoint extensions of the operators on intervals $(-\infty,-\epsilon)$, 
$[-\epsilon, \epsilon]$ and $(\epsilon,\infty)$, we consider two distinct representations of $H$ in the Hilbert space $L^{2}(\R)\ot 
\C^{2}$. The first one, denoted $A$, is obtained 
by integrating $H$ in the interval $[-\epsilon,\epsilon]\ni 0$ and then taking the limit $\epsilon\downarrow0$; this gives the required 
boundary condition in defining the domain $D(A)$ of $A$. The second representation of $H$ is a distribution 
$B=H_{0}+\gamma\delta\ot  I$ on $W_{0}^{2}(\R\backslash\{0\})\ot \C^{2}$, with $\delta$ the delta-function. Here and elsewhere, 
$W_{0}^{p}$, with $p=1,2$, is the closure of $C_{0}^{\infty}$ in $W^{p}$, the Sobolev space of functions whose (weak) derivatives of 
order $\leq p$ are in $L^{2}$ \cite[Sec.~3]{Adams03}; we also use the notation $\RO\equiv\R\backslash\{0\}$. By default, we take
into account the isomorphism from $L^{2}(\R)\ot \C^{2}$ to $L^{2}(\R;\C^{2})$ by \citealt[Theorem~II.10]{Reed80}.

To demonstrate that representatives $A$ and $B$ are proper realizations of $H$ we explore the method developed by \cite{Coutinho09}.
As a result, we establish that $[A,A_{0}]=0$ in a strict (classical) sense, and that $[B,B_{0}]=0$ in a weak (distributional) sense.
Here $B_{0}=(U+V_{F}\ot  I)\upharpoonright W_{0}^{1}(\R\backslash\{0\})\ot \C^{2}$. The commutator predetermines a nonempty set of 
common eigenfunctions of $A$ and $A_{0}$, provided $\Omega,\eta>0$ (Theorem~\ref{thm:EA}). The latter inequality shows that extra 
states in $\disc(A)$ can be observed only for nonzero spin-orbit and Raman coupling, and that their appearance in the spectrum is
essentially dependent on the location of the dressed spin states \citep{Lin11} in the dispersion curve.

Although $A$ and $B$ are equivalent representations for providing the spectral characteristics for $H$ in $L^{2}(\R)\ot \C^{2}$, 
we explore both of them. The main reason for such a choice is because the interaction is drawn in $B$ explicitly, and thus one can
easier attach the physical meaning to $B$, rather than $A$; the same applies to $B_{0}$ and $A_{0}$, respectively. On the other hand, 
equivalence classes of functions in $\mrm{ker}(\lambda\ot  I-B)$, with $\lambda\in\disc(B)$, are in a one-to-one
correspondence with functions in $\mrm{ker}(\lambda\ot  I-A)$, with the same $\lambda$, if and only if one imposes certain conditions 
on the normalization constant and the eigenfunction itself (Sec.~\ref{sec:spectrum}). This agrees with \citealt[Sec.~V.4]{Reed80}, which 
in our case says that weak solutions $\mrm{ker}(\lambda\ot  I-B)$ are equal to the classical solutions 
$\mrm{ker}(\lambda\ot  I-A)$ if and only if the classical solutions exist.

The paper is organized as follows. In Sec.~\ref{sec:operators}, we give basic definitions of potential $V$ and the representatives
$A$, $B$, and examine their correctness. Sec.~\ref{sec:fermi} deals mainly with operator $A_{0}$ and its distributional version $B_{0}$. 
As a result, the Fermi pseudopotential $V_{F}$ is introduced. In Sec.~\ref{sec:extra}, we provide spin-orbit induced states for $A$,
as well as compute the essential spectrum. Finally, we compute the remaining part of the discrete spectrum of $A$ ($B$) in 
Sec.~\ref{sec:spectrum}, and summarize the results in Sec.~\ref{sec:summary}.

\section{Preliminaries}\label{sec:operators}

Throughout, we define $\RO\equiv\R\backslash\{0\}$, $\Lp{X}\equiv L^{2}(X)\ot \C^{2}$, $\Wp{p}{X}\equiv W^{p}(X)\ot \C^{2}$
for $p=1,2$, $\CO{X}^{2}\equiv C_{0}^{\infty}(X)\ot \C^{2}$ for some $X\subseteq\R$,
$\Sigma\equiv[-\epsilon,\epsilon]$ for some $\epsilon>0$.

Given function $V$ which is defined as the limit of a sequence of rectangles

\begin{equation}
V(x)=\gamma v(x)\quad(\gamma\in\RO),\quad 
v(x)=\left\{\begin{array}{ll}1/(2\epsilon), & x\in \Sigma, \\ 0, & x\in \R\backslash\Sigma \end{array}\right.\quad\text{as}\quad 
\epsilon\downarrow0.
\label{eq:V}
\end{equation}

\noindent{}Then $v$ is supported in $\Sigma$, and it approaches $\delta$, the delta-function, in the usual sense of distributions,
with the property $\int_{-\infty}^{\infty}v(x)dx=1$. As a matter of fact, $v$ has a wider meaning than $\delta$ in the sense that
\cite[Eq.~(7)]{Coutinho09}

\begin{subequations}\label{eq:V-delta}
\begin{align}
&\int_{-\infty}^{\infty}v(x)f(x)dx=f(0)+\frac{1}{2}\lim_{\epsilon\downarrow0}\sum_{n=1}^{\infty}
\frac{\epsilon^{n}}{(n+1)!}\left(f^{(n)}(0_{+})+(-1)^{n}f^{(n)}(0_{-}) \right), \nonumber \\ 
&f(0_{\pm})\equiv\lim_{\epsilon\downarrow0}f(\pm\epsilon),\quad
f(0)\equiv(f(0_{+})+f(0_{-}))/2\quad 
(f\in C_{0}^{\infty}(\RO)) \label{eq:V-delta-1}
\end{align}

\noindent{}($f^{(n)}$ is the $n$th derivative of $f$ with respect to $x\in\R$ at a given point). As a functional, $v(f)\equiv f(0)$
if and only if $f^{(n)}(\pm\epsilon)\propto\epsilon^{-s(n)}$ for $s(n)<n$ for $n=1,2,\ldots$

In particular, (\ref{eq:V-delta-1}) yields

\begin{equation}
\int_{-\infty}^{\infty}v(x)f(x)dx=f(0)+\lim_{\epsilon\downarrow0}\sum_{n=1}^{\infty}
\frac{\epsilon^{2n}f^{(2n)}(0)}{(2n+1)!}\quad (f\in C_{0}^{\infty}(\R)).
\label{eq:V-delta-2}
\end{equation}
\end{subequations}

\noindent{}Equation~(\ref{eq:V-delta-2}) serves for the criterion in establishing whether the delta-function approximation of 
(\ref{eq:V}) is a proper one. This is done by calculating $f^{(n)}$ at $x=0$ for all $n=0,1,\ldots$, where function $f$ is in the
kernel of the operator that involves $V$ as in (\ref{eq:V}). Afterward, one needs to verify under what circumstances the infinite series
in (\ref{eq:V-delta-2}) converges. For the analysis of specific operator classes, the reader is referred to
\cite{Coutinho09,Griffiths99}. The application of (\ref{eq:V-delta-2}) to $H$ in (\ref{eq:Hamilt}) is examined below.

Let $f\in\mrm{ker}\:H$ in $\Sigma$. The solutions $f(x)\sim e^{kx}$ ($k\in\C;x\in\Sigma$) are found by solving the characteristic equation
for $H$: $\det[(H_{0}+\gamma/(2\epsilon))e^{kx}]=0$ ($\gamma\in\RO$) or explicitly,

$$
k^{4}+(\eta^{2}-\gamma/\epsilon)k^{2}-(\Omega^{2}-\gamma^{2}/\epsilon^{2})/4=0\quad(\eta,\Omega\geq0;\gamma\in\RO;\epsilon>0).
$$

\noindent{}The solutions with respect to $k\in\C$ read

\begin{equation}
k_{ss^{\prime}}=\frac{s^{\prime}}{\sqrt{2}}\left((\gamma/\epsilon)-\eta^{2}+s
\sqrt{\eta^{4}-2\eta^{2}(\gamma/\epsilon)+\Omega^{2}} \right)^{1/2}\quad(s,s^{\prime}=\pm1),
\label{eq:kss}
\end{equation}

\noindent{}and so

$$
k_{ss^{\prime}}\to s^{\prime}k/\sqrt{\epsilon}\quad(k=\sqrt{\gamma/2}\in\C;s^{\prime}=\pm1)
\quad\text{as}\quad\epsilon\downarrow0.
$$

\noindent{}The upper, $f_{1}$, and lower, $f_{2}$, components of $f$ are then of the form

\begin{equation}
f_{1}(x)=\sum_{ss^{\prime}}a_{ss^{\prime}}e^{k_{ss^{\prime}}x},\quad
f_{2}(x)=\sum_{ss^{\prime}}b_{ss^{\prime}}e^{k_{ss^{\prime}}x}\quad(x\in\Sigma)
\label{eq:fss}
\end{equation}

\noindent{}for some $\{a_{ss^{\prime}}\in\C\co s,s^{\prime}=\pm1\}$, $\{b_{ss^{\prime}}\in\C\co s,s^{\prime}=\pm1\}$.
Clearly,

$$
f_{1}(\pm\epsilon)=\sum_{ss^{\prime}}a_{ss^{\prime}}e^{\pm s^{\prime}k\sqrt{\epsilon}}\to\sum_{ss^{\prime}}a_{ss^{\prime}},\quad 
f_{2}(\pm\epsilon)=\sum_{ss^{\prime}}b_{ss^{\prime}}e^{\pm s^{\prime}k\sqrt{\epsilon}}\to\sum_{ss^{\prime}}b_{ss^{\prime}}
$$

\noindent{}as $\epsilon\downarrow0$. Hence $f(0_{+})=f(0_{-})$, $f\in\mrm{ker}\:H$ is continuous at $x=0$.

The $n$th derivative ($n=0,1,\ldots$) of $f$ at $x=0$ is found by differentiating $f(x)\in\CO{\Sigma}^{2}$ $n$ times with respect to $x$
and then setting $x=0$, 

$$
f_{1}^{(n)}(0)=k^{n}\epsilon^{-n/2}\sum_{ss^{\prime}}(s^{\prime})^{n}a_{ss^{\prime}},\quad 
f_{2}^{(n)}(0)=k^{n}\epsilon^{-n/2}\sum_{ss^{\prime}}(s^{\prime})^{n}b_{ss^{\prime}}\quad(\epsilon>0).
$$

\noindent{}As seen, $f^{(n)}(0)\propto\epsilon^{-s(n)}$ with $s(n)=n/2<n$ for $n=1,2,\ldots$. But then 
$\epsilon^{2n}f^{(2n)}(0)\propto\epsilon^{n}\to0$ as $\epsilon\downarrow0$, and the infinite series in (\ref{eq:V-delta-2}) vanishes.
This proves that, as a functional, $v(f)\equiv f(0)$ makes sense for functions in certain domains of $H$. 

As a result, at least two possibilities are valid to construct these domains. The first one is obtained by integrating $Hf$ in $\Sigma$ 
and then taking the limit $\epsilon\downarrow0$. In agreement with (\ref{eq:V-delta-2}) and the discussion above, this gives the operator

\begin{align}
A=H_{0},\quad D(A)=&\Biggl\{f=\left(\begin{matrix}f_{1} \\ f_{2} \end{matrix}\right)\in \Wp{2}{\RO}\co 
\gamma f(0)=f^{\prime}(0_{+})-f^{\prime}(0_{-}) \nonumber \\
&+(i\eta\ot \sigma_{2})(f(0_{+})-f(0_{-})),H_{0}f\in\Lp{\R} \Biggr\} \label{eq:A}
\end{align}

\noindent{}($\gamma\in\RO;\eta\geq0$) where $f(0)$ is of the form in (\ref{eq:V-delta-1}). It appears from (\ref{eq:A}) that for zero 
spin-orbit coupling $\eta=0$, or continuous functions at $x=0$, the boundary condition in $D(A)$ is a familiar relation valid for the 
operators with $\delta$-interaction. This suggests the second realization of $H$ in $\Lp{\R}$, namely,

\begin{equation}
B=(H_{0}+\gamma\delta\ot  I)\upharpoonright\WpO{2}{\RO}\quad(\gamma\in\RO)
\label{eq:B}
\end{equation}

\noindent{}with $\delta$ the delta-function. Here we recall that although $B$ is a distribution, operator $A$ can be interpreted
in the classical sense due to the fact \cite[Theorem~3.17]{Adams03} that distributional and classical derivatives coincide whenever
the latter exist (and certainly are continuous on $\RO$).

If, however, we start from the pure point-interaction (that is, $\delta$-interaction) and integrate $B$ in $\Sigma$, we derive that the 
property $f(0_{+})=f(0_{-})$ is only the (additional, though reasonable) assumption, as also discussed by \cite{Coutinho97}. Moreover, 
the operator $H\upharpoonright\WpO{2}{\Sigma}$ is not self-adjoint, and it has deficiency indices, d.i., (2,2) as $\epsilon\downarrow0$. 
This means that additional boundary conditions at $\pm\epsilon$ are required, and so again, $f(0_{+})$ is not necessarily equal to 
$f(0_{-})$, in general. This is our motive to inspect the boundary condition in $D(A)$ in its most general form.

To this end, let us comment on the self-adjointness of operator $A$ ($B$).

Let us solve $H_{0}f_{z}=zf_{z}$ for some $z\in\C\backslash\R$. The solutions $f_{z}$ are of the form (\ref{eq:fss}), with
$k_{ss^{\prime}}$ in (\ref{eq:kss}) replaced by 

\begin{equation}
k_{ss^{\prime}}=\frac{s^{\prime}}{\sqrt{2}}\left(2z-\eta^{2}+s
\sqrt{\eta^{4}-4\eta^{2}z+\Omega^{2}} \right)^{1/2}\quad(s,s^{\prime}=\pm1;\eta,\Omega\geq0).
\label{eq:kss2}
\end{equation}

\noindent{}For $x>0$, one requires $\mrm{Re}\:k_{ss^{\prime}}<0$ in order to make solutions square integrable. This yields 
$(s,s^{\prime})=(1,-1)$ and $(-1,-1)$. For $x<0$, however, $\mrm{Re}\:k_{ss^{\prime}}>0$, and possible values are $(s,s^{\prime})=(1,1)$ 
and $(-1,1)$. Evidently, the intersection of possible solutions which are square integrable in the whole $\R$ is the empty set. In
terms of deficiency indices, operator $A$ has d.i. (0,0), hence self-adjoint.

A general solution to $Bf_{z}=zf_{z}$ for $z\in\C\backslash\R$ can be written in the form

\begin{align}
f_{z}(x)=&-\frac{\gamma}{2\pi}\int_{-\infty}^{\infty}dp\:\frac{e^{ipx}}{\Delta_{z}(p)}
((p^{2}-z)\ot  I-\hat{U}(p))f(0), \nonumber \\
\Delta_{z}(p)=&(p^{2}-z)^{2}-\eta^{2}p^{2}-(\Omega/2)^{2}\quad(\eta,\Omega\geq0), \label{eq:fp}
\end{align}

\noindent{}where $\hat{U}(p)=\eta p\ot \sigma_{2}+(\Omega/2)\ot \sigma_{3}$ is the Fourier transform of $U$. To see this,
one simply needs to solve $\widehat{(Bf)}(p)=z\hat{f}(p)$ ($p\in\R$) by noting that 
$\widehat{(Bf)}(p)=\hat{H}_{0}(p)\hat{f}(p)+\gamma f(0)$, in agreement with (\ref{eq:B}); here $\hat{H}_{0}(p)=p^{2}\ot  I+
\hat{U}(p)$. It follows from (\ref{eq:fp}) that the
Fourier transform $\hat{f}_{z}$ of $f_{z}$ is proportional to $p^{-2}$. As a result, $p^{2}\hat{f}_{z}$ is not in $\Lp{\R}$
\cite[Sec.~IX.6]{Reed75}, hence $B$ has d.i. (0,0). Similarly to the case for the Dirac operator, one can also construct the
quadratic form $\gamma\vert f(0)\vert^{2}$ and show that it satisfies the KLMN theorem \cite[Theorem~X.17]{Reed75} with respect to
$H_{0}\upharpoonright \WpO{2}{\RO}$.

\section{Fermi pseudopotential}\label{sec:fermi}

In the present section we consider the operator

\begin{align}
A_{0}=U,\quad D(A_{0})=&
\Biggl\{f=\left(\begin{matrix}f_{1} \\ f_{2} \end{matrix}\right)\in\Wp{1}{\RO}\co 
\gamma f(0)=f^{\prime}(0_{+})-f^{\prime}(0_{-}) \nonumber \\
&+(i\eta\ot \sigma_{2})(f(0_{+})-f(0_{-})),Uf\in\Lp{\R} \Biggr\} \label{eq:A0}
\end{align}

\noindent{}($\gamma\in\RO;\eta\geq0$). As discussed in Sec.~\ref{sec:intro} of the present paper, $U$ has a meaning of the atom-light
coupling originated from the synthetic gauge fields (for more details, the reader is referred to \cite{Dalibard11}). Now we wish to 
examine the properties of its representative $A_{0}$.

The arguments of self-adjointness are similar to those for operator $A$ in the previous section. One solves $Uf_{z}=zf_{z}$
with respect to $f_{z}=\left(\begin{matrix}f_{1,z} \\ f_{2,z}\end{matrix}\right)$ for $z\in\C\backslash\R$, and gets that

\begin{align}
&f_{1,z}(x)=c_{1}\cosh(\omega_{z}x)+c_{2}\sqrt{\frac{\Omega+2z}{\Omega-2z}}\sinh(\omega_{z}x), \nonumber \\
&f_{2,z}(x)=c_{2}\cosh(\omega_{z}x)+c_{1}\sqrt{\frac{\Omega-2z}{\Omega+2z}}\sinh(\omega_{z}x) \label{eq:fz12}
\end{align}

\noindent{}($c_{1},c_{2}\in\C;x\in\RO;\omega_{z}=\sqrt{\Omega^{2}-4z^{2}}/(2\eta);\Omega\geq0;\eta>0$). Clearly, $f_{z}$ is not in 
$\Lp{\R}$, hence $A_{0}$ has d.i. (0,0). [Alternatively, one can explore the Weyl's criterion by noting from (\ref{eq:f12}) that there
is one solution in $L^{2}$ as $x\to\infty$, and one solution as $x\to-\infty$.]

The boundary condition in (\ref{eq:A0}) suggests that, similarly to the case of operator $A$ and its distributional version $B$,
there should be some weak form, $B_{0}$, of $A_{0}$ as well.

Given $B_{0}=U+V_{F}\ot  I$ on $\WpO{1}{\RO}$ for some distribution $V_{F}$. Let us integrate $(U+V_{F}\ot  I)f$ in $\Sigma$ 
for $f\in D(A_{0})$, and then take the limit $\epsilon\downarrow0$,

\begin{align}
0=&\int_{-\epsilon}^{\epsilon}(U+V_{F}\ot  I)f(x)dx=
-(i\eta\ot \sigma_{2})(f(0_{+})-f(0_{-})) \nonumber \\
&+\int_{-\epsilon}^{\epsilon}(V_{F}\ot  I)f(x)dx
\Longrightarrow
\int_{-\epsilon}^{\epsilon}(V_{F}\ot  I)f(x)dx=(i\eta\ot \sigma_{2})(f(0_{+})-f(0_{-})) \nonumber \\
=&\gamma f(0)-(f^{\prime}(0_{+})-f^{\prime}(0_{-})). \label{eq:VFderiv}
\end{align}

In \citep{Coutinho04}, the authors have defined the modified $\delta^{\prime}$-interaction to which we refer as the
$\delta_{p}^{\prime}$-interaction,

\begin{equation}
\delta_{p}^{\prime}(f)=\delta^{\prime}(\tilde{f}),\quad\text{with}\quad
\tilde{f}(x)=\left\{\begin{array}{ll}f(x)-(f(0_{+})-f(0_{-}))/2, & x>0, \\ 
f(x)+(f(0_{+})-f(0_{-}))/2, & x<0. \end{array}\right.
\label{eq:deltap}
\end{equation}

\noindent{}The reason for modifying the original $\delta^{\prime}$-interaction is that it is not applicable to discontinuous
functions, as pointed out by \cite{Coutinho97}. The integral \cite[Eq.~(44)]{Coutinho97}

$$
\int_{-\epsilon}^{\epsilon}\delta^{\prime}(x)f(x)dx=-\frac{1}{2}(f^{\prime}(0_{+})+f^{\prime}(0_{-}))-
\frac{1}{2\alpha}(f(0_{+})-f(0_{-}))\quad(0<\alpha<\epsilon)
$$

\noindent{}diverges for discontinuous functions, as $\epsilon\downarrow0$, because of the last term. On the other hand
(see also \cite[Eq.~(24)]{Coutinho04}), the integral

\begin{align*}
&\int_{-\epsilon}^{\epsilon}\delta_{p}^{\prime}(x)f(x)dx=\int_{-\epsilon}^{\epsilon}\delta^{\prime}(x)\tilde{f}(x)dx=
-\int_{-\epsilon}^{\epsilon}\delta(x)\tilde{f}^{\prime}(x)=-\frac{1}{2}(f^{\prime}(0_{+})+f^{\prime}(0_{-}))
\end{align*}

\noindent{}is convergent. Below we show that the divergent term can be canceled in the following manner:

\begin{prop}\label{prop:1}
Let $f\in C^{1}(\RO)$. Let $\delta_{p}^{\prime}$ be as in (\ref{eq:deltap}). Then for $\epsilon\downarrow0$,

\begin{equation}
\int_{-\epsilon}^{\epsilon}(\delta_{p}^{\prime}(x_{-})-\delta_{p}^{\prime}(x_{+}))f(x)dx=
\int_{-\epsilon}^{\epsilon}(\delta^{\prime}(x_{-})-\delta^{\prime}(x_{+}))f(x)dx
=f^{\prime}(0_{-})-f^{\prime}(0_{+})
\label{eq:VFderiv2}
\end{equation}

\noindent{}where $\delta_{p}^{\prime}(x_{\pm})=\delta_{p}^{\prime}(x\pm\alpha)$ for $0<\alpha<\epsilon$, and the same for
$\delta^{\prime}(x_{\pm})$.
\end{prop}

\begin{proof}
To prove the statement we only need the definition of $\delta_{p}^{\prime}$, (\ref{eq:deltap}), and that of $\delta^{\prime}$,
\cite{Coutinho97,Griffiths93},

\begin{equation}
\delta^{\prime}(x)=\lim_{\beta\downarrow0}\frac{1}{2\beta}(\delta(x+\beta)-\delta(x-\beta)).
\label{eq:delta}
\end{equation}

Let $0<\beta<\alpha<\epsilon$ and $\alpha+\beta<\epsilon$ for $\epsilon>0$ arbitrarily small. By (\ref{eq:delta}),

\begin{align*}
&\int_{-\epsilon}^{\epsilon}(\delta^{\prime}(x-\alpha)-\delta^{\prime}(x+\alpha))f(x)dx=\frac{1}{2\beta}
\int_{-\epsilon}^{\epsilon}[(\delta(x-\alpha+\beta)-\delta(x-\alpha-\beta)) \\
&-(\delta(x+\alpha+\beta)-\delta(x+\alpha-\beta))]f(x)dx
=\frac{1}{2\beta}[(f(\alpha-\beta)-f(\alpha+\beta)) \\
&-(f(-\alpha-\beta)-f(-\alpha+\beta))]=-\frac{f(\alpha+\beta)-f(\alpha-\beta)}{2\beta}+
\frac{f(-\alpha+\beta)-f(-\alpha-\beta)}{2\beta} \\
&=-f^{\prime}(\alpha)+f^{\prime}(-\alpha).
\end{align*}

\noindent{}In the limit $\alpha\downarrow0$, this gives (\ref{eq:VFderiv2}).

By (\ref{eq:deltap}) and (\ref{eq:delta}),

\begin{align*}
&\int_{-\epsilon}^{\epsilon}(\delta_{p}^{\prime}(x-\alpha)-\delta_{p}^{\prime}(x+\alpha))f(x)dx=
\int_{-\epsilon}^{\epsilon}(\delta^{\prime}(x-\alpha)-\delta^{\prime}(x+\alpha))\tilde{f}(x)dx \\
&=\frac{1}{2\beta}\int_{-\epsilon}^{\epsilon}[(\delta(x-\alpha+\beta)-\delta(x-\alpha-\beta))-
(\delta(x+\alpha+\beta)-\delta(x+\alpha-\beta))]\tilde{f}(x)dx \\
&=\frac{1}{2\beta}[(\tilde{f}(\alpha-\beta)-\tilde{f}(\alpha+\beta))-
(\tilde{f}(-\alpha-\beta)-\tilde{f}(-\alpha+\beta))]=\frac{1}{2\beta}\Biggl[\Biggl(f(\alpha-\beta) \\
&-\frac{f(0_{+})-f(0_{-})}{2}-f(\alpha+\beta)+\frac{f(0_{+})-f(0_{-})}{2}\Biggr)-\Biggl(f(-\alpha-\beta)+\frac{f(0_{+})-f(0_{-})}{2} \\
&-f(-\alpha+\beta)-\frac{f(0_{+})-f(0_{-})}{2}\Biggr)\Biggr]=
-\frac{f(\alpha+\beta)-f(\alpha-\beta)}{2\beta}+\frac{f(-\alpha+\beta)-f(-\alpha-\beta)}{2\beta} \\
&=-f^{\prime}(\alpha)+f^{\prime}(-\alpha).
\end{align*}

\noindent{}In the limit $\alpha\downarrow0$, we again derive (\ref{eq:VFderiv2}). The proof is accomplished.
\end{proof}

We apply Proposition~\ref{prop:1} to functions in $D(A_{0})$. Then the substitution of the left-hand side of (\ref{eq:VFderiv2}) in
(\ref{eq:VFderiv}) along with $\int_{-\epsilon}^{\epsilon}\delta(x)f(x)dx=f(0)$ ($f(0)$ as in (\ref{eq:V-delta-1})) yields

\begin{equation}
B_{0}=(U+V_{F}\ot  I)\upharpoonright\WpO{1}{\RO},\quad
V_{F}(x)=\gamma\delta(x)+\delta^{\prime}(x_{-})-\delta^{\prime}(x_{+})
\label{eq:VF}
\end{equation}

\noindent{}($\gamma,x\in\RO$), with $\delta^{\prime}(x_{-})-\delta^{\prime}(x_{+})$ relevant to Proposition~\ref{prop:1}.

By virtue of (\ref{eq:VF}) we have found that suitably rotated in spin space (recall the unitary operator
$I\ot  e^{-i\theta\sigma_{3}}$, with $\theta\equiv 3\pi/4\!\!\!\mod\pi$, discussed in Sec.~\ref{sec:intro}), the operator $A_{0}/\eta$, 
with $A_{0}$ as in (\ref{eq:A0}) and the spin-orbit coupling $\eta>0$, describes the Dirac-like (or Weyl--Dirac) particle
of spin one-half and mass $\Omega/(2\eta)$ moving in the Fermi pseudopotential $V_{F}/\eta$.

We close the present section with the spectral properties of $A_{0}$ ($B_{0}$).

\begin{thm}\label{thm:A0}
\begin{enumerate}[\upshape (i)]
\item 
The resolvent of $A_{0}$ is given by

\begin{align*}
(R_{z}(A_{0})f)(x)=&\int_{-\infty}^{\infty}dx^{\prime}\:(A_{0}-z\ot  I)^{-1}
(x-x^{\prime})f(x^{\prime})
\\
\intertext{($f\in\Lp{\R}\cap L^{1}(\R)^{2}$), with the integral kernel (Green's function)}
(A_{0}-z\ot  I)^{-1}(x-x^{\prime})=&\frac{2\eta^{2}\omega_{z}(A_{0}^{0}-z\ot  I)^{-1}(x-x^{\prime}) }{
(\gamma z+2\omega_{z}(\eta^{2}+z))^{2}-(\Omega/2)^{2}(\gamma+2\omega_{z})^{2} } \\
&\times [2\eta^{2}\omega_{z}\ot  I-(\gamma+2\omega_{z})((\Omega/2)\ot \sigma_{3}-z\ot  I)] \\
\intertext{($x\neq x^{\prime};x,x^{\prime}\in\R;z\in\C\backslash\sigma(A_{0});\Omega,\eta>0;
\mrm{Re}\:\omega_{z}\neq0;\gamma\in\RO$), where $A_{0}^{0}=U\upharpoonright W_{0}^{1}(\RO)^{2}$ and}
(A_{0}^{0}-z\ot  I)^{-1}(x-x^{\prime})=&
\frac{e^{-\omega_{z}\vert x-x^{\prime}\vert}\ot  I}{2\eta^{2}\omega_{z}}\:
(i\eta\omega_{z}\:\mrm{sgn}(x-x^{\prime})\ot \sigma_{2}+(\Omega/2)\ot \sigma_{3}+z\ot  I)
\end{align*}

\noindent{}($x\neq x^{\prime};x,x^{\prime}\in\R;z\in \C\backslash\sigma(A_{0}^{0});\Omega,\eta>0;
\mrm{Re}\:\omega_{z}\neq0$), where $\omega_{z}$ is as in (\ref{eq:fz12});

\item 

\begin{align*}
\disc(A_{0})=&\Bigl\{-\Omega/2<\varepsilon<\Omega/2\co \gamma/2+\omega\pm\eta\sqrt{(\Omega\mp2\varepsilon)/(\Omega\pm2\varepsilon)}=0; \\
&\omega=\sqrt{\Omega^{2}-4\varepsilon^{2}}/(2\eta);
\gamma<0;\Omega,\eta>0\Bigr\},\quad\text{with the eigenfunctions} \nonumber \\
f(x)=& f(0)e^{-\omega \vert x\vert}+(\Theta(-x)e^{\omega x}-\Theta(x)e^{-\omega x})
\left(\begin{matrix} f_{2}(0)\sqrt{\frac{\Omega+2\varepsilon}{\Omega-2\varepsilon}} \\ 
f_{1}(0)\sqrt{\frac{\Omega-2\varepsilon}{\Omega+2\varepsilon}} \end{matrix}\right)
\end{align*}

\noindent{}($x\in\RO;\Omega,\eta>0;\vert\varepsilon\vert<\Omega/2$), where $\Theta$ denotes the Heaviside theta function, 
and $f_{2}(0)=0$ ($f_{1}(0)=0$) for the upper (lower) sign in $\disc(A_{0})$;

\item 
$\disc(B_{0})=\disc(A_{0})$, with $\mrm{ker}(\varepsilon\ot  I-B_{0})$ ($\varepsilon\in\disc(B_{0})$) containing equivalence classes of
functions $f(x)=-(\gamma+2\omega)(A_{0}^{0}-\varepsilon\ot  I)^{-1}(x)f(0)$ ($x\in\RO;\gamma<0;\omega>0$);

\item $\ess(A_{0})=\ess(B_{0})=\sigma(A_{0}^{0})=(-\infty,-\Omega/2]\cup[\Omega/2,\infty)$ ($\Omega\geq0$);

\item 
There are no eigenvalues embedded into the essential spectrum: $\disc(A_{0})\cap\ess(A_{0})=\varnothing$.
\end{enumerate}
\end{thm}

\begin{figure}[htp]
\centering
\includegraphics[scale=0.50]{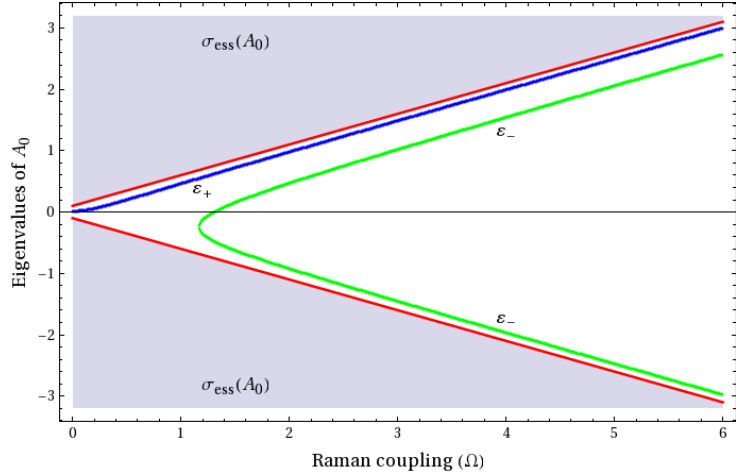}
\caption{ 
(Color online) Computed spectrum of operator $A_{0}$ (see Eq.~(\ref{eq:A0}) and Theorem~\ref{thm:A0}) for the 
point-interaction strength $\gamma=-1$ and the spin-orbit-coupling strength $\eta=0.6$ (in $\hbar=c=1$ units). The eigenvalues divided by 
$\eta>0$ are those of the one-dimensional Dirac-like operator for the particle of spin one-half and mass $\Omega/(2\eta)$ moving in the 
Fermi pseudopotential (\ref{eq:VF}). In figure, red lines show the border of the essential spectrum of $A_{0}$, which is
$\pm\Omega/2$. The blue $\varepsilon_{+}$ (green $\varepsilon_{-}$) line, showing the bound state as a function of the Raman coupling 
$\Omega>0$, corresponds to the eigenfunction with a zero-valued lower (upper) component at the origin $x=0$. 
}
\label{fig:EigenvaluesA0}
\end{figure}

\begin{rem}
(1) In order to find the eigenvalues $\varepsilon\in\disc(A_{0})$ explicitly, one needs to solve the cubic equation with respect to 
$\varepsilon$, as it is seen from Theorem~\ref{thm:A0}-(ii). The solutions to such type of equations are well known for a long time. 
However, their general form is rather complicated and we did not find it valuable here. Instead of that we displayed the spectrum of 
$A_{0}$ versus the Raman coupling $\Omega>0$ in Fig.~\ref{fig:EigenvaluesA0}.

(2) We also note that, unlike in Theorem~\ref{thm:A0}-(iii), where $f(0)$ is undetermined because of the delta-function, $f(0)$ in 
Theorem~\ref{thm:A0}-(ii) obeys the form as in (\ref{eq:V-delta-1}). The solutions in $\mrm{ker}(\varepsilon\ot  I-A_{0})$ are strict 
so that $f(0)$ can be replaced by any constant ($1$, say).
\end{rem}

\begin{proof}[Proof of Theorem~\ref{thm:A0}]

(i) The integral kernel $(A_{0}-z\ot  I)^{-1}(x)$ (for simplicity, we replace $x-x^{\prime}$ by $x$) is defined through the 
formal differential equation

$$
(U+V_{F}\ot  I-z\ot  I)G_{0}(x;z)=\delta(x)\ot  I.
$$

\noindent{}In agreement with (\ref{eq:VF}), $G_{0}(x;z)$ is of the form

\begin{equation}
G_{0}(x;z)=\frac{1}{2\pi}\int_{-\infty}^{\infty}dp\:e^{ipx}\hat{G}_{0}(p;z),\quad
\hat{G}_{0}(p;z)=\hat{G}_{0}^{0}(p;z)\Phi(\gamma;z),\quad
\hat{G}_{0}^{0}(p;z)=\frac{z\ot  I+\hat{U}(p)}{\eta^{2}(p^{2}+\omega_{z}^{2})}
\label{eq:Green}
\end{equation}

\noindent{}and $\Phi(\gamma;z)=I\ot  I-\gamma G_{0}(0;z)-G_{0}^{\prime}(0_{-};z)+G_{0}^{\prime}(0_{+};z)$. As one would have noticed, 
$\hat{G}_{0}^{0}(p;z)$ is the Fourier transform of $(A_{0}^{0}-z\ot  I)^{-1}(x)$.
Recalling that the integrals $\int_{-\infty}^{\infty}dp\:e^{ipx}/(p^{2}+\omega^{2})=(\pi/\omega)e^{-\omega\vert x\vert}$,
$\int_{-\infty}^{\infty}dp\:pe^{ipx}/(p^{2}+\omega^{2})=i\pi\:\mrm{sgn}(x)e^{-\omega\vert x\vert}$ for $x\in\RO$ and 
$\mrm{Re}\:\omega\neq0$, we derive the expression

\begin{equation}
G_{0}(x;z)=(A_{0}^{0}-z\ot  I)^{-1}(x)\Phi(\gamma;z),
\label{eq:GR0}
\end{equation}

\noindent{}with the integral kernel $(A_{0}^{0}-z\ot  I)^{-1}(x)$ as in the theorem. By using this equation, calculate 
$G_{0}(0;z)=(G_{0}(0_{+};z)+G_{0}(0_{-};z))/2$ and $G_{0}^{\prime}(0_{\pm};z)$, and get the equation for $\Phi(\gamma;z)$,

$$
[2\eta^{2}\omega_{z}\ot  I+(\gamma+2\omega_{z})((\Omega/2)\ot \sigma_{3}+z\ot  I)]\Phi(\gamma;z)=
2\eta^{2}\omega_{z}\ot  I.
$$

\noindent{}Substitute obtained expression of $\Phi(\gamma;z)$ in (\ref{eq:GR0}), replace $x$ by $x-x^{\prime}$ back again and get (i), as 
required. Note that $f\in L^{1}(\R)^{2}$ is because of $(B_{0}-z\ot  I)R_{z}(A_{0})=I\ot  I$ (in the sense of distributions), that 
is, the resolvent of $A_{0}$ ($B_{0}$) is a distribution, and hence the equation $(A_{0}-z\ot  I)R_{z}(A_{0})=I\ot  I$ is 
meaningless in the classical sense.

(ii) The discrete spectrum is easily recovered by setting the denominator of the resolvent of $A_{0}$ equal to zero. As for the
eigenfunctions, we begin with (\ref{eq:fz12}) by letting $z\equiv\varepsilon\in\disc(A_{0})$ and $\omega_{\varepsilon}\equiv\omega$. We 
rewrite (\ref{eq:fz12}) in the following form

\begin{align}
&f_{1}(x)=\tfrac{1}{2}\Theta(x)e^{-\omega x}\left(c_{1}-c_{2}\sqrt{\frac{\Omega+2\varepsilon}{\Omega-2\varepsilon}}\right)+
\tfrac{1}{2}\Theta(-x)e^{\omega x}\left(c_{1}+c_{2}\sqrt{\frac{\Omega+2\varepsilon}{\Omega-2\varepsilon}}\right), \nonumber \\
&f_{2}(x)=\tfrac{1}{2}\Theta(x)e^{-\omega x}\left(c_{2}-c_{1}\sqrt{\frac{\Omega-2\varepsilon}{\Omega+2\varepsilon}}\right)+
\tfrac{1}{2}\Theta(-x)e^{\omega x}\left(c_{2}+c_{1}\sqrt{\frac{\Omega-2\varepsilon}{\Omega+2\varepsilon}}\right) \label{eq:f12}
\end{align}

\noindent{}($\vert\varepsilon\vert<\Omega/2;\Omega,\eta>0;\omega>0$), where $f_{z\equiv\varepsilon}\equiv f$, and 
$f_{j,\varepsilon}\equiv f_{j}$ for $j=1,2$. By (\ref{eq:f12}), with $f(0)$ as in (\ref{eq:V-delta-1}),

$$
f(0)=\frac{1}{2}\left(\begin{matrix}c_{1} \\ c_{2} \end{matrix}\right),\quad
f(0_{+})-f(0_{-})=-\left(\begin{matrix}c_{2}\sqrt{\frac{\Omega+2\varepsilon}{\Omega-2\varepsilon}} \\ 
c_{1}\sqrt{\frac{\Omega-2\varepsilon}{\Omega+2\varepsilon}} \end{matrix}\right),\quad
f^{\prime}(0_{+})-f^{\prime}(0_{-})=-\omega\left(\begin{matrix}c_{1} \\ c_{2} \end{matrix}\right).
$$

\noindent{}But $f\in D(A_{0})$, (\ref{eq:A0}), and so it must hold

\begin{equation}
c_{1}\left(\frac{\gamma}{2}+\omega+\eta \sqrt{\frac{\Omega-2\varepsilon}{\Omega+2\varepsilon}}\right)=0,\quad
c_{2}\left(\frac{\gamma}{2}+\omega-\eta \sqrt{\frac{\Omega+2\varepsilon}{\Omega-2\varepsilon}}\right)=0
\label{eq:vareps}
\end{equation}

\noindent{}($\vert\varepsilon\vert<\Omega/2;\Omega,\eta>0;\omega>0$), with $c_{j}=2f_{j}(0)$ for $j=1,2$. After some elementary
simplifications, equations (\ref{eq:f12}) and (\ref{eq:vareps}) lead to (ii).

(iii) Let $f\in\mrm{ker}(\varepsilon\ot  I-B_{0})$ for some $\varepsilon\in\R$. Combining the Fourier transform of (\ref{eq:VF})
with (\ref{eq:Green}) we get that

$$
f(x)=-(A_{0}^{0}-\varepsilon\ot  I)^{-1}(x)\widehat{(V_{F}\ot  I)f},\quad
\widehat{(V_{F}\ot  I)f}=\gamma f(0)+f^{\prime}(0_{-})-f^{\prime}(0_{+}).
$$

\noindent{}Now, if we calculate $\gamma f(0)+f^{\prime}(0_{-})-f^{\prime}(0_{+})$ by taking $f$ from the left side of the above 
expressions, we get that $\widehat{(V_{F}\ot  I)f}=(\gamma+2\omega)f(0)$ and that

$$
\left(I\ot  I+\frac{\gamma+2\omega}{2\eta^{2}\omega}\left((\Omega/2)\ot \sigma_{3}+
\varepsilon\ot  I\right) \right)f(0)=0\quad(\eta>0;\omega>0)
$$

\noindent{}thus yielding (iii).

(iv) The essential spectrum of $A_{0}$ is found from the dispersion curve $\varepsilon(p)$ which in turn is found by taking
the Fourier transform of $U$ and solving the eigenvalue equation, namely,

$$
\det\left(\begin{matrix}\Omega/2-\varepsilon(p) & -i\eta p \\ i\eta p & -\Omega/2-\varepsilon(p) \end{matrix}\right)=0.
$$ 

\noindent{}The result reads $\varepsilon(p)=\pm\sqrt{(\Omega/2)^{2}+(\eta p)^{2}}$ for all $p\in\R$. 

The essential spectrum of $B_{0}$ is found from the integral kernel of the resolvent of $A_{0}$, by virtue of (iii). This is exactly
the case as for deriving the spectrum of $A_{0}^{0}$. Then one needs to solve $p^{2}+\omega_{z}^{2}=0$ with respect to 
$z\equiv\varepsilon(p)$ ($p\in\R$). The solutions are those as above, and hence (iv) holds.

(v) The present item immediately follows from (iv) and from the requirement that, for $\varepsilon\in\disc(A_{0})$, it holds
$-\Omega/2<\varepsilon<\Omega/2$.
\end{proof}

\section{Spin-orbit coupling induced states}\label{sec:extra}

\begin{lem}\label{lem:commutator}
We have: 

\begin{enumerate}
\item $[A,A_{0}]=0$ on $D(A)$ strictly;

\item $[B,B_{0}]=0$ almost everywhere in $\R\ot \C^{2}$.
\end{enumerate}
\end{lem}

\begin{proof}
We note that $\Wp{p}{\RO}\subset\Wp{p^{\prime}}{\RO}$ for $p>p^{\prime}$; see eg \cite[p.~276]{Herczynski89}.
By (\ref{eq:A}) and (\ref{eq:A0}), $D(A)\subset D(A_{0})$. By (\ref{eq:B}) and (\ref{eq:VF}), $D(B)\subset D(B_{0})$.
Thus $[A,A_{0}]$ makes sense since $R(A_{0})\cap D(A)\subset R(A_{0})\cap D(A_{0})=D(A_{0})$, 
$R(A)\cap D(A_{0})\supset D(A)\cap D(A_{0})=D(A)$, and the same for $[B,B_{0}]$ ($R$ is the range).

Item (1) is easy to perform: $[A,A_{0}]$ on $D(A)$ is given by
$[H_{0},U]=[-\Delta\ot  I,U]=0$. The same applies to the resolvents $R_{z_{0}}(A_{0})$, $R_{z}(A)$ 
($z_{0},z\in\C$) and to the exponents $e^{itA_{0}}$, $e^{isA}$ ($t,s\in\R$) in consonance with 
\cite[Theorem~VIII.13]{Reed80}. The fact that the exponents commute follows from the commutation relation of resolvents. This can be 
seen by noting eg $R_{z_{0}}(A_{0})=i\int_{0}^{\infty}dt\:e^{-it(B_{0}-z_{0}\otimes  I)}$ ($\mrm{Im}\:z_{0}>0$). That the resolvents 
commute (weakly), the easiest way to see this is to apply (\ref{eq:Green}) and (\ref{eq:Green2}), where one concludes that the integral
$\int_{-\infty}^{\infty}([R_{z}(A),R_{z_{0}}(A_{0})]f)(x)dx$ is equal to 
$\int_{-\infty}^{\infty}[\hat{G}(0;z),\hat{G}_{0}(0;z_{0})]f(x)dx=0$, provided $f\in L^{1}(\R)^{2}$.

In order to prove (2), we integrate $[B,B_{0}]$ in the interval $X\subseteq\RO$ because $D(B)\subset D(B_{0})$ contains functions
which are well-defined for $x\in\RO$. In this case, all integrands containing $\delta$ or $\delta^{\prime}$ (see (\ref{eq:delta})) vanish
because the argument of $\delta$ ($\delta^{\prime}$) is nonzero for all $x\in X$. The remaining terms, that is, those which do not
include deltas, commute with each other.
Finally, we extend $X\subseteq\RO$ to the whole $\R$ by setting $X=(-\infty,-\epsilon)\cup(\epsilon,\infty)$ as 
$\epsilon\downarrow0$, and we have (2).
\end{proof}

We already know from Theorem~\ref{thm:A0}-(ii) that $\mrm{ker}(\varepsilon\ot  I-A_{0})\subset D(A_{0})$ is a nonempty set for 
$\varepsilon\in\disc(A_{0})$. Now, we assume that $\disc(A)\neq\varnothing$, and let $\lambda\in\disc(A)$. Then by 
Lemma~\ref{lem:commutator}, 

\begin{equation}
D(A)\supset\mrm{ker}(\varepsilon\ot  I-A_{0})\cap\mrm{ker}(\lambda\ot  I-A)
\equiv\mrm{ker}(\lambda(\varepsilon)\ot  I-A)
\label{eq:RAA0}
\end{equation}

\noindent{}for some $\lambda(\varepsilon)\in \so(A)\subset\disc(A)$. We say that the set $\so(A)$ contains
spin-orbit coupling induced states $\lambda(\varepsilon)$. This is because $\so(A)$ is nonempty only for nonzero spin-orbit coupling
$\eta>0$, in agreement with Theorem~\ref{thm:A0}.

Here, our main goal is to establish $\so(A)$. For that reason we prove that:

\begin{lem}\label{lem:ess}
Let $A$ and $B$ be as in (\ref{eq:A}) and (\ref{eq:B}). Then:

\begin{enumerate}[\upshape (i)]
\item 
The resolvent of $A$ is given by

\begin{align*}
(R_{z}(A)f)(x)=&\int_{-\infty}^{\infty}dx^{\prime}\:(A-z\ot  I)^{-1}(x-x^{\prime})f(x^{\prime})
\\
\intertext{($f\in\Lp{\R}\cap L^{1}(\R)^{2}$), with the integral kernel (Green's function)}
(A-z\ot  I)^{-1}(x-x^{\prime})=&2p_{1}p_{2}(p_{1}+p_{2})(A^{0}-z\ot  I)^{-1}(x-x^{\prime}) \\
&\times\frac{ p_{1}p_{2}(i\gamma+2(p_{1}+p_{2}))\ot  I-i\gamma((\Omega/2)\ot \sigma_{3}-z\ot  I ) }{ 
(2p_{1}p_{2}(p_{1}+p_{2})+i\gamma(p_{1}p_{2}+z))^{2}+(\gamma\Omega/2)^{2} } \\
\intertext{($x\neq x^{\prime};x,x^{\prime}\in\R;z\in\C\backslash\sigma(A);\Omega,\eta\geq0;\gamma\in\R;\mrm{Im}\:p_{j}>0;j=1,2$), 
where $A^{0}=H_{0}\upharpoonright W_{0}^{2}(\RO)^{2}$, and the integral kernel of $A^{0}$ is given by}
(A^{0}-z\ot  I)^{-1}(x-x^{\prime})=&
\frac{i}{2(p_{1}^{2}-p_{2}^{2})}\Biggl(\frac{e^{ip_{1}(x-x^{\prime})}}{p_{1}}
(p_{1}^{2}\ot  I-z\ot  I-\hat{U}(p_{1})) \\
&-\frac{e^{ip_{2}(x-x^{\prime})}}{p_{2}}
(p_{2}^{2}\ot  I-z\ot  I-\hat{U}(p_{2})) \Biggr)\quad(x>x^{\prime}), \\
=&
\frac{i}{2(p_{1}^{2}-p_{2}^{2})}\Biggl(\frac{e^{-ip_{1}(x-x^{\prime})}}{p_{1}}
(p_{1}^{2}\ot  I-z\ot  I-\hat{U}(-p_{1})) \\
&-\frac{e^{-ip_{2}(x-x^{\prime})}}{p_{2}}
(p_{2}^{2}\ot  I-z\ot  I-\hat{U}(-p_{2})) \Biggr)\quad(x<x^{\prime})
\end{align*}

\noindent{}($x,x^{\prime}\in\R;z\in\C\backslash\sigma(A^{0});\Omega,\eta\geq0;\mrm{Im}\:p_{j}>0;j=1,2$),
$p_{1,2}=s_{1,2}\sqrt{ z+\eta^{2}/2\pm(1/2)\sqrt{ \eta^{2}(\eta^{2}+4z)+\Omega^{2} } }$ ($s_{j}=\pm1;j=1,2$);

\item 
$\ess(A)=\ess(B)=\sigma(A^{0})=[J(\eta,\Omega),\infty)$, where 
$J(\eta,\Omega)$ is equal to $\lambda_{0}\equiv-[\eta^{2}+(\Omega/\eta)^{2}]/4$ for $0\leq\Omega\leq\eta^{2}$, and 
to $-\Omega/2$ for $\Omega>\eta^{2}\geq0$.
\end{enumerate}
\end{lem}

\begin{proof}
(i) The proof is pretty much similar to that of (\ref{eq:fp}) and Theorem~\ref{thm:A0}-(i). The integral kernel 
$(A-z\ot  I)^{-1}(x)$ (for simplicity, we replace $x-x^{\prime}$ by $x$) is defined through the formal differential equation

$$
(-\Delta\ot  I+U+\gamma\delta(x)\ot  I-z\ot  I)G(x;z)=\delta(x)\ot  I.
$$

\noindent{}Then 

\begin{equation}
G(x;z)=\frac{1}{2\pi}\int_{-\infty}^{\infty}dp\:e^{ipx}\hat{G}(p;z),\quad
\hat{G}(p;z)=\hat{G}^{0}(p;z)\Psi(\gamma;z),\quad
\hat{G}^{0}(p;z)=\frac{(p^{2}-z)\ot  I-\hat{U}(p)}{\Delta_{z}(p)},
\label{eq:Green2}
\end{equation}

\noindent{}with $\Delta_{z}(p)$ as in (\ref{eq:fp}) and $\Psi(\gamma;z)=I\ot  I-\gamma G(0;z)$. As one would have noticed,
$\hat{G}^{0}(p;z)$ is the Fourier transform of $(A^{0}-z\ot  I)^{-1}(x)$. For more convenience, we rewrite the denominator 
by $\Delta_{z}(p)=(p^{2}-p_{1}^{2})(p^{2}-p_{2}^{2})$, with $p_{j}$ ($j=1,2$) as in Lemma~\ref{lem:ess}-(i).

Without loss of generality, we assume that $\mrm{Im}\:p_{j}>0$ ($j=1,2$). Then the integration over $p\in\R$ can be performed in two 
distinct ways. Consider

$$
\varphi(\zeta)=
\frac{e^{i\zeta x}((\zeta^{2}-z)\ot  I-\eta \zeta\ot \sigma_{2}-(\Omega/2)\ot \sigma_{3})}{
(\zeta^{2}-p_{1}^{2})(\zeta^{2}-p_{2}^{2})}\quad(x\in\RO;z,\zeta\in\C),
$$

\noindent{}and integrate it around the contour $\mathcal{C}$ oriented counterclockwise, with the poles $p_{1}$, $p_{2}$. This implies
that the integral exists for $x>0$. Similarly, integrate $\varphi(\zeta)$ around the contour $\mathcal{C}^{\prime}$
oriented counterclockwise but with the poles $-p_{1}$, $-p_{2}$, and get $x<0$ for the existence of the integral. 
[We note that these two contours of integration are not unique. One can choose, for example, the contour with poles $p_{1}$, 
$-p_{2}$ ($\mrm{Im}\:p_{1}>0;\mrm{Im}\:p_{2}<0$) so that the integral exists for $x>0$, and the contour with poles $-p_{1}$, $p_{2}$
(again, $\mrm{Im}\:p_{1}>0;\mrm{Im}\:p_{2}<0$) so that the integral exists for $x<0$.]

By the residue theorem,

\begin{align*}
&\int_{-\infty}^{\infty}\varphi(p)dp+\lim_{R\to\infty}\int_{\mathcal{C}}\varphi(\zeta)d\zeta=
2\pi i\underset{\zeta=p_{1},p_{2}}{\mrm{res}}\:\varphi(\zeta), \\
&-\int_{-\infty}^{\infty}\varphi(p)dp+\lim_{R\to\infty}\int_{\mathcal{C}^{\prime}}\varphi(\zeta)d\zeta=
2\pi i\underset{\zeta=-p_{1},-p_{2}}{\mrm{res}}\:\varphi(\zeta),
\end{align*}

\noindent{}where the contour integration is performed over $\zeta=Re^{i\psi}$ ($\psi\in[0,\pi]$) in the first contour, and over 
$\zeta=Re^{i\psi}$ ($\psi\in[\pi,2\pi]$) in the second contour. In the limit $R\to\infty$, function $\vert\varphi(\zeta)\vert\to0$ for 
$x>0$ in the first integral, and for $x<0$ in the second one. 

The residues are easy to calculate by noting that 

$$
\frac{1}{(z^{2}-p_{1}^{2})(z^{2}-p_{2}^{2})}=\frac{1}{2(p_{1}^{2}-p_{2}^{2})}
\left(\frac{1}{p_{1}(z-p_{1})}-\frac{1}{p_{1}(z+p_{1})}-\frac{1}{p_{2}(z-p_{2})}+\frac{1}{p_{2}(z+p_{2})} \right).
$$

\noindent{}After some elementary simplifications, and replacing $x$ with $x-x^{\prime}$, we obtain the integral kernel 
of the resolvent of $A^{0}$ as in Lemma~\ref{lem:ess}-(i).

Following (\ref{eq:Green2}),

\begin{equation}
G(x;z)=(A^{0}-z\ot  I)^{-1}(x)\Psi(\gamma;z).
\label{eq:GR}
\end{equation}

\noindent{}By using this equation, calculate $G(0;z)=(G(0_{+};z)+G(0_{-};z))/2$ and get the equation for $\Psi(\gamma;z)$,

$$
[(i\gamma+2(p_{1}+p_{2}))\ot  I+(i\gamma/(p_{1}p_{2}))((\Omega/2)\ot \sigma_{3}+z\ot  I)]\Psi(\gamma;z)=
2(p_{1}+p_{2})\ot  I.
$$

\noindent{}Substitute obtained expression of $\Psi(\gamma;z)$ in (\ref{eq:GR}), replace $x$ with $x-x^{\prime}$ and get 
the resolvent of $A$ as required. That $f\in L^{1}(\R)^{2}$, the arguments are those as in the proof of Theorem~\ref{thm:A0}-(i).

\begin{figure}[htp]
\centering
\includegraphics[scale=0.50]{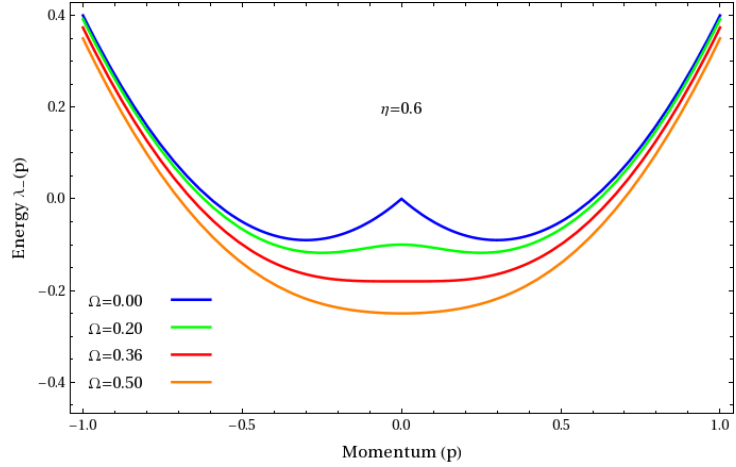}
\caption{ 
(Color online) Computed lower branch of dispersion in (\ref{eq:dispersion}) for the spin-orbit-coupling strength $\eta=0.6$
(in $\hbar=c=1$ units), for a range of Raman couplings $\Omega\geq0$. As $\Omega$ increases ($\Omega>\eta^{2}$), the two dressed
spin states \citep{Lin11} are merged into a single minimum $-\Omega/2$ at $p=0$. This is a regime when the spin-orbit coupling induced
states $\so(A)$, Theorem~\ref{thm:EA}, are observed below the continuous spectrum as well as above it. For $\Omega\leq\eta^{2}$, the spin 
states have two minima $-[\eta^{2}+(\Omega/\eta)^{2}]/4$ at $p=\pm\sqrt{\eta^{4}-\Omega^{2}}/(2\eta)$, and the spin-orbit induced
states are embedded into the essential spectrum of $A$.
}
\label{fig:Dispersion}
\end{figure}

(ii) The essential spectrum of $A$ as well as the spectrum of $A^{0}$ is found from (\ref{eq:Green2}) by solving $\Delta_{z}(p)=0$ 
($p\in\R$) with respect to $z\equiv\lambda(p)$, whereas for $B$, one needs to solve the same equation due to (\ref{eq:fp}). The solutions 
read

\begin{equation}
\lambda_{\pm}(p)=p^{2}\pm\sqrt{\eta^{2}p^{2}+(\Omega/2)^{2}}\geq \lambda_{-}(p).
\label{eq:dispersion}
\end{equation}

\noindent{}The lower bound of $\lambda_{\pm}(p)$ is found by differentiating $\lambda_{-}(p)$ with respect to $p\in\R$. One finds three
critical points: $p_{1}=0$, $p_{2}=-\sqrt{\eta^{4}-\Omega^{2}}/(2\eta)$ and $p_{3}=\sqrt{\eta^{4}-\Omega^{2}}/(2\eta)$. As seen,
$p_{2}$ and $p_{3}$ are in $\R$ only for $\Omega\leq\eta^{2}$. Hence it holds $\lambda_{\pm}(p)\geq-[\eta^{2}+(\Omega/\eta)^{2}]/4$. If, 
however, $\Omega>\eta^{2}$, only $p_{1}$ is valid. Then $\lambda_{\pm}(p)\geq-\Omega/2$. This proves that $\ess(A)=\ess(B)$, hence
(ii), and the proof of the statement is accomplished.
\end{proof}

\begin{rem}
For the illustrative and comparison purposes (see \cite[Fig.~1b]{Lin11} and \cite[Fig.~2c]{Galitski13}), we displayed the dispersion 
relation $\lambda_{-}(p)$, (\ref{eq:dispersion}), in Fig.~\ref{fig:Dispersion}.
\end{rem}

We are now in a position to establish the properties of spin-orbit coupling induced states.

\begin{thm}\label{thm:EA}
Given $A$ as in (\ref{eq:A}) and $A_{0}$ as in (\ref{eq:A0}). Then:

\begin{enumerate}[\upshape (i)]
\item 
$\disc(A)\supset\so(A)=
\bigl\{\varepsilon-\omega^{2}\co \varepsilon\in\disc(A_{0})\backslash\bigl\{-\eta^{2}/2,
\Omega/2-\eta^{2}\bigr\};\Omega,\eta>0\bigr\}$;
\item 
$\so(A)=\sigma_{<}(A)\cup\sigma_{>}(A),\quad\sigma_{>}(A)=\sigma_{1}(A)\cup\sigma_{2}(A)$;
\item 
$\sigma_{<}(A)=\bigl\{\lambda(\varepsilon)\in\so(A)\co \varepsilon\in\disc(A_{0});-\Omega/2<\varepsilon<\Omega/2-\eta^{2};
\Omega>\eta^{2}>0\bigr\}$;
\item 
$\sigma_{1}(A)=\bigl\{\lambda(\varepsilon)\in\so(A)\co \varepsilon\in\disc(A_{0});\Omega/2-\eta^{2}<\varepsilon<\Omega/2;
\Omega>\eta^{2}>0\bigr\}$;
\item 
$\sigma_{2}(A)=\bigl\{\lambda(\varepsilon)\in\so(A)\co \varepsilon\in\disc(A_{0});0<\Omega\leq\eta^{2}\bigr\}$;
\item 
$\so(A)\cap\ess(A)=\sigma_{2}(A)$ for $0<\Omega\leq\eta^{2}$;
\item 
$\so(A)\cap\ess(A)=\sigma_{1}(A)$ for $\Omega>\eta^{2}>0$;
\item 
$\so(B)=\so(A)$. The equivalence classes of functions from the kernel $\mrm{ker}(\lambda(\varepsilon)\ot  I-B)$, for
$\lambda(\varepsilon)\in\so(B)$, are of the form given in Theorem~\ref{thm:A0}-(iii).
\end{enumerate}

\noindent{}The eigenfunctions that correspond to $\lambda(\varepsilon)\in \so(A)$ are as in Theorem~\ref{thm:A0}-(ii).
\end{thm}

\begin{proof}
The proof is essentially based on the combination of Theorem~\ref{thm:A0} with Lemmas~\ref{lem:commutator}--\ref{lem:ess}.

(i) In agreement with Lemma~\ref{lem:commutator}-(1), and in particular (\ref{eq:RAA0}), substitute 
$f\in\mrm{ker}(\varepsilon\ot  I-A_{0})$ (refer to Theorem~\ref{thm:A0}-(ii)) in $\mrm{ker}(\lambda(\varepsilon)\ot  I-A)$ for some 
$\lambda(\varepsilon)\in\R$. Then

\begin{align*}
0=&
f_{1}(0)\left(-\omega^{2}+\frac{\Omega}{2}-\lambda(\varepsilon)-\omega\eta\sqrt{\frac{\Omega-2\varepsilon}{\Omega+2\varepsilon}} \right)
\\
&+
f_{2}(0)
\left(\mp\left(-\omega^{2}+\frac{\Omega}{2}-\lambda(\varepsilon)\right)\sqrt{\frac{\Omega+2\varepsilon}{\Omega-2\varepsilon}}\pm
\omega\eta \right), \\
0=&
f_{1}(0)
\left(\mp\left(-\omega^{2}-\frac{\Omega}{2}-\lambda(\varepsilon)\right)\sqrt{\frac{\Omega-2\varepsilon}{\Omega+2\varepsilon}}\mp
\omega\eta \right) \\
&+
f_{2}(0)\left(-\omega^{2}-\frac{\Omega}{2}-\lambda(\varepsilon)+\omega\eta\sqrt{\frac{\Omega+2\varepsilon}{\Omega-2\varepsilon}} \right)
\end{align*}

\noindent{}($\omega$ as in Theorem~\ref{thm:A0}), where the upper sign corresponds to $x>0$, and the lower one to $x<0$.
It appears from above that for either $f_{2}(0)=0$ or $f_{1}(0)=0$, the following holds,

\begin{align*}
0=&-\omega^{2}+\frac{\Omega}{2}-\lambda(\varepsilon)-\omega\eta\sqrt{\frac{\Omega-2\varepsilon}{\Omega+2\varepsilon}}, \\
0=&-\omega^{2}-\frac{\Omega}{2}-\lambda(\varepsilon)+\omega\eta\sqrt{\frac{\Omega+2\varepsilon}{\Omega+2\varepsilon}}.
\end{align*}

\noindent{}The solution $\lambda(\varepsilon)$ satisfying the above system of equations is given by 
$\lambda(\varepsilon)=\varepsilon-\omega^{2}$ or explicitly, $\varepsilon-(\Omega^{2}-4\varepsilon^{2})/(4\eta^{2})$.

In order to accomplish the proof of (i), it remains to establish valid eigenvalues $\varepsilon$ from $\disc(A_{0})$ thus 
generating proper eigenvalues $\lambda(\varepsilon)$ from $\so(A)$.

By a straightforward inspection, $\lambda_{0}\leq\lambda(\varepsilon)<\Omega/2$ for all $\Omega,\eta>0$, where $\lambda_{0}$ is as in
Lemma~\ref{lem:ess}-(ii). The lower bound is obtained at $\varepsilon=-\eta^{2}/2$ (the solution to $d\lambda(\varepsilon)/d
\varepsilon=0$). On the other hand, $\lambda_{0}\leq-\Omega/2$ and $\lambda(\varepsilon)=-\Omega/2$ at $\varepsilon=\Omega/2-\eta^{2}$
($\varepsilon=-\Omega/2$ is improper due to Theorem~\ref{thm:A0}-(ii)). Therefore, the points $\varepsilon=-\eta^{2}/2$ and
$\Omega/2-\eta^{2}$, which hold whenever $\Omega>\eta^{2}>0$, must be excluded as the resonant states, by Theorem~\ref{thm:A0}-(i)
(inspect solutions to $\omega_{z}=0$ with respect to $z$ given by $\pm\Omega/2$) and by Lemma~\ref{lem:ess}-(i) (inspect solutions to 
$p_{1}^{2}=p_{2}^{2}$ with respect to $z$ given by $\lambda_{0}$, and solutions to $p_{j}=0$, $j=1,2$, given by $\pm\Omega/2$).
Item (i) holds.

(ii)--(v) The reason for extracting $\so(A)$ into subsets is in different behavior of the involved eigenvalues:
$\sup\sigma_{<}(A)=\inf\ess(A)$ and $\inf\sigma_{>}(A)=\inf\ess(A)$. This is easy to verify by considering $\lambda(\varepsilon)$
and $J(\eta,\Omega)$: For $0<\Omega\leq\eta^{2}$, one finds that $\lambda(\varepsilon)>J(\eta,\Omega)$, which is $\sigma_{2}(A)$.
For $\Omega>\eta^{2}>0$, $\lambda(\varepsilon)<J(\eta,\Omega)$ for $-\Omega/2<\varepsilon<\Omega/2-\eta^{2}$, thus yielding
$\sigma_{<}(A)$, and $\lambda(\varepsilon)>J(\eta,\Omega)$ for $\Omega/2-\eta^{2}<\varepsilon<\Omega/2$, thus yielding
$\sigma_{1}(A)$. The values $\lambda(\varepsilon)=J(\eta,\Omega)$ are excluded due to the previous discussion (these are resonant states).

(vi) Since $J(\eta,\Omega)=\lambda_{0}$ for $0<\Omega\leq\eta^{2}$, we have that $\so(A)=\sigma_{2}(A)$ in this regime. But
$\inf\sigma_{2}(A)=\inf\ess(A)$, and hence (vi) holds.

(vii) For $\Omega>\eta^{2}>0$, $J(\eta,\Omega)=-\Omega/2$. In the present regime we have that $\so(A)=\sigma_{1}(A)$ with
$\inf\sigma_{1}(A)=-\Omega/2$. This gives (vii).

(viii) Following Lemma~\ref{lem:commutator}-(2), we need to show that (weak) solutions in $\mrm{ker}(\lambda(\varepsilon)\ot  I-B)$
yield eigenvalues $\lambda(\varepsilon)\in\so(B)=\so(A)$. By Theorem~\ref{thm:A0}-(iii),

\begin{subequations}\label{eq:12}
\begin{equation}
0=\int_{-\infty}^{\infty}(B-\lambda(\varepsilon)\ot  I)f(x)dx=\int_{-\infty}^{\infty}(H_{0}f)(x)dx+
(\gamma-2\lambda(\varepsilon)/\omega)f(0),
\end{equation}

\noindent{}where we have explored the integral $\int_{-\infty}^{\infty}f(x)dx=(2/\omega)f(0)$ for $\omega>0$
(recall $f\in L^{1}(\R)^{2}$ in Theorem~\ref{thm:A0}-(i) and Lemma~\ref{lem:ess}-(i)). But

\begin{align}
\int_{-\infty}^{\infty}(H_{0}f)(x)dx=&
-\int_{-\infty}^{\infty}f^{\prime\prime}(x)dx-(i\eta\ot \sigma_{2})\int_{-\infty}^{\infty}f^{\prime}(x)dx
 \nonumber \\
&+((\Omega/2)\ot \sigma_{3})\int_{-\infty}^{\infty}f(x)dx=((\Omega/\omega)\ot \sigma_{3})f(0),
\end{align}
\end{subequations}

\noindent{}and hence the combination of (\ref{eq:12}) yields

\begin{equation}
((\Omega/2)\ot \sigma_{3}+(\gamma\omega/2-\lambda(\varepsilon))\ot  I)f(0)=0.
\label{eq:3}
\end{equation}

\noindent{}Equation~(\ref{eq:3}) has solutions with respect to $\lambda(\varepsilon)\in\R$ only if either $f_{2}(0)=0$ or
$f_{1}(0)=0$ (recall Theorem~\ref{thm:A0}). Then it holds $\lambda(\varepsilon)=(\gamma\omega\pm\Omega)/2$, where the upper sign is for 
$f_{2}(0)=0$, and the lower one for $f_{1}(0)=0$. Recalling that $\omega=\sqrt{(\Omega/2)^{2}-\varepsilon^{2}}/\eta$, we recover 
$\so(A)$. This accomplishes the proof of the theorem.
\end{proof}

The points in $\so(A)\subset\disc(A)$ are illustrated in Fig.~\ref{fig:EigenvaluesEA}.

\begin{figure}[htp]
\centering
\includegraphics[scale=0.50]{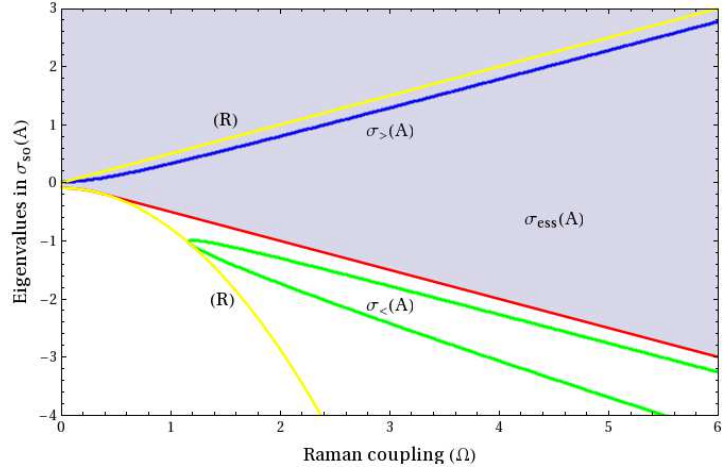}
\caption{ 
(Color online) Computed spin-orbit coupling induced states $\so(A)\subset\disc(A)$ (refer to Theorem~\ref{thm:EA}) for the 
point-interaction strength $\gamma=-1$ and the spin-orbit-coupling strength $\eta=0.6$ (in $\hbar=c=1$ units). In figure, red line 
shows the border $\inf\ess(A)$ of the essential spectrum of $A$ (Lemma~\ref{lem:ess}). The eigenvalues $\lambda(\varepsilon)\in 
\so(A)$ ($\varepsilon\in\disc(A_{0})$), as functions of the Raman coupling $\Omega>0$, are drawn by the blue ($\sigma_{>}(A)$) and green 
($\sigma_{<}(A)$) lines. Resonant states of $A$ are drawn by yellow curves (R).
}
\label{fig:EigenvaluesEA}
\end{figure}

\section{Discrete spectrum}\label{sec:spectrum}

As yet, we have established the part of $\disc(A)$ which is associated with discontinuous eigenfunctions at $x=0$. These states
originate from the property that $A$ commutes with $A_{0}$, where $A_{0}/\eta$ ($\eta>0$) is unitarily equivalent to the one-dimensional
Dirac operator for the particle in Fermi pseudopotential.

In this section, our main goal is to determine the remaining part of $\disc(A)$, namely, $\disc(A)\backslash\so(A)$, thus
recovering all discrete states of the spin-orbit Hamiltonian, and to show that the associated eigenfunctions are continuous in the whole 
$\R$.

\begin{thm}\label{thm:A}
Let $A$ and $B$ be as in (\ref{eq:A}) and (\ref{eq:B}), respectively. Then:

\begin{enumerate}
\item 

\begin{align*}
\disc(A)=&\disc(B)=
\bigl\{\lambda<-\Omega/2\co 2p_{1}p_{2}(p_{1}+p_{2})+i\gamma(p_{1}p_{2}+\lambda\pm\Omega/2)=0; \\
&\lambda\neq\lambda_{0};\Omega\geq0;\eta>0;\gamma<0;\mrm{Im}\:p_{j}>0;j=1,2 \bigr\}\bigcup \so(A),
\end{align*}

\noindent{}where $\so(A)$ is given in Theorem~\ref{thm:EA}, the $p_{j}$ ($j=1,2$) and $\lambda_{0}$ are as in Lemma~\ref{lem:ess}, 
with $s_{1}=+1$, $s_{2}=\pm1$, $z\equiv\lambda$; 
\item The equivalence classes of functions from $\mrm{ker}(\lambda\ot  I-B)$ (with $\lambda\in\disc(B)\backslash\so(B)$) are of the 
form $-\gamma(A^{0}-\lambda\ot  I)^{-1}(x)f(0)$ (with $x\in\RO;\gamma<0$), with the integral kernel, for $z\equiv\lambda$, as in 
Lemma~\ref{lem:ess}-(i);

\item The (strict) solutions $\mrm{ker}(\lambda\ot  I-A)$ associated with $\lambda$ from $\disc(A)\backslash\so(A)$ are of the form: 
	\begin{enumerate}
	\item For $\lambda\in \disc(A)\backslash\so(A)$ with the upper sign,

\begin{subequations}\label{eq:AAAA}
\begin{align}
f(x)=&C\left[\frac{e^{ip_{1}x}}{p_{1}}\left(\begin{matrix}\lambda+\Omega/2-p_{1}^{2} \\
i\eta p_{1}\end{matrix}\right)-\frac{e^{ip_{2}x}}{p_{2}}\left(\begin{matrix}\lambda+\Omega/2-p_{2}^{2} \\
i\eta p_{2}\end{matrix}\right) \right]\quad(x>0), \label{eq:AAAA-a} \\
=&C\left[\frac{e^{-ip_{1}x}}{p_{1}}\left(\begin{matrix}\lambda+\Omega/2-p_{1}^{2} \\
-i\eta p_{1}\end{matrix}\right)-\frac{e^{-ip_{2}x}}{p_{2}}\left(\begin{matrix}\lambda+\Omega/2-p_{2}^{2} \\
-i\eta p_{2}\end{matrix}\right) \right]\quad(x<0) \label{eq:AAAA-b}
\end{align}
\end{subequations}

\noindent{}for any $C\in\C\backslash\{0\}$, $\eta>0$;	
	\item For $\lambda\in \disc(A)\backslash\so(A)$ with the lower sign,

\begin{subequations}\label{eq:AAAA1}
\begin{align}
f(x)=&C\left[\frac{e^{ip_{1}x}}{\lambda+\Omega/2-p_{1}^{2}}\left(\begin{matrix}\lambda+\Omega/2-p_{1}^{2} \\
i\eta p_{1}\end{matrix}\right)-\frac{e^{ip_{2}x}}{\lambda+\Omega/2-p_{2}^{2}}\left(\begin{matrix}\lambda+\Omega/2-p_{2}^{2} \\
i\eta p_{2}\end{matrix}\right) \right] \nonumber \\
&(x>0), \label{eq:AAAA1-a} \\
=&C\left[-\frac{e^{-ip_{1}x}}{\lambda+\Omega/2-p_{1}^{2}}\left(\begin{matrix}\lambda+\Omega/2-p_{1}^{2} \\
-i\eta p_{1}\end{matrix}\right)+\frac{e^{-ip_{2}x}}{\lambda+\Omega/2-p_{2}^{2}}\left(\begin{matrix}\lambda+\Omega/2-p_{2}^{2} \\
-i\eta p_{2}\end{matrix}\right) \right] \nonumber \\
&(x<0) \label{eq:AAAA1-b}
\end{align}
\end{subequations}

\noindent{}for any $C\in\C\backslash\{0\}$, $\eta>0$;	
	\item For $\eta=0$, we have that the discrete spectrum is given by the union
	$\disc(A)\backslash\so(A)=\disc(A)=\bigl\{-\gamma^{2}/4\pm\Omega/2\co \gamma<-2\sqrt{\Omega}\bigr\}\cup\bigl\{
	-\gamma^{2}/4-\Omega/2\co -2\sqrt{\Omega}<\gamma<0\bigr\}$; the associated 
	eigenfunctions are $C\chi_{\pm}e^{\gamma\vert x\vert/2}$, with $\sigma_{3}\chi_{\pm}=\pm\chi_{\pm}$ 
	($x\in\RO;C\in\C\backslash\{0\};\Omega\geq0;\gamma<0$);
	\end{enumerate}
\item There are no eigenvalues from $\disc(A)\backslash\so(A)$ embedded into the essential spectrum of $A$:
$(\disc(A)\backslash\so(A))\cap\ess(A)=\varnothing$.
\end{enumerate}
\end{thm}

\begin{rem}
(1) As is seen from the theorem, the eigenfunctions of $A$ and $B$, which correspond to the upper sign for $\lambda$ in 
$\disc(A)\backslash\so(A)$, coincide if and only if

\begin{subequations}\label{eq:unique}
\begin{equation}
f_{1}(0)\equiv f_{1}(0_{+})=f_{1}(0_{-})=\frac{-2iC}{\gamma}(p_{1}^{2}-p_{2}^{2}),\quad
f_{2}(0)\equiv f_{2}(0_{+})=f_{2}(0_{-})=0
\end{equation}

\noindent{}($C\in\C\backslash\{0\};\gamma<0$). The eigenfunctions of $A$ and $B$, which correspond to the lower sign for $\lambda$ in 
$\disc(A)\backslash\so(A)$, coincide if and only if

\begin{equation}
f_{1}(0)\equiv f_{1}(0_{+})=f_{1}(0_{-})=0,\quad
f_{2}(0)\equiv f_{2}(0_{+})=f_{2}(0_{-})=\frac{2C}{\gamma\eta}(p_{1}^{2}-p_{2}^{2})
\end{equation}
\end{subequations}

\noindent{}($C\in\C\backslash\{0\};\gamma<0;\eta>0$). 

Therefore, equations~(\ref{eq:unique}) provide unique solutions (up to the constant $C$) for functions $f_{j}(0)$ ($j=1,2$) which are 
undetermined in $\mrm{ker}(\lambda\ot  I-B)$; see Theorem~\ref{thm:A}-(2).

(2) It is interesting to compare the eigenfunctions at $x=0$ (having the meaning as in (\ref{eq:V-delta-1})), which correspond to
the spin-orbit coupling induced states (Theorem~\ref{thm:EA}), with those given above. For $\lambda(\varepsilon)\in\so(A)$ with the
upper sign, $f_{2}(0_{+})=-f_{2}(0_{-})$ yields $f_{2}(0)=0$; in comparison, $f_{2}(0)\equiv f_{2}(0_{+})=f_{2}(0_{-})=0$ for
$\lambda\in\disc(A)\backslash\so(A)$ with the upper sign. Hence in both cases, the <<total>> lower component $f_{2}(0)=0$. Similarly, 
there is also another case but with the upper component $f_{1}(0)=0$.

(3) As in Theorem~\ref{thm:EA}, the eigenvalues $\lambda$ in $\disc(A)\backslash\so(A)$ can be written in an explicit form by
solving the cubic equation. We chose not to do that, but displayed $\lambda$ graphically instead; see Fig.~\ref{fig:EigenvaluesA}.
\end{rem}

\begin{figure}[htp]
\centering
\includegraphics[scale=0.50]{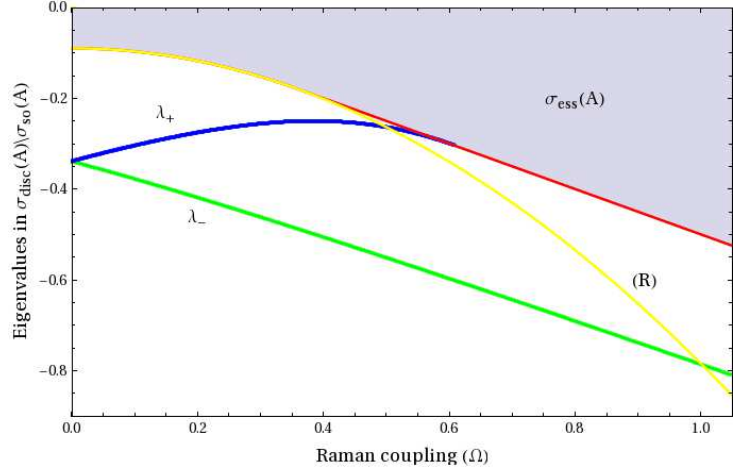}
\caption{ 
(Color online) The eigenvalues of $A$ associated with everywhere continuous eigenfunctions.
The point-interaction strength $\gamma=-1$ and the spin-orbit-coupling strength $\eta=0.6$ (in $\hbar=c=1$ units). In figure, red 
line shows the border $\inf\ess(A)$ of the essential spectrum of $A$ (Lemma~\ref{lem:ess}). The blue $\lambda_{+}$ (green 
$\lambda_{-}$) line, showing the bound state as a function of the Raman coupling $\Omega\geq0$, corresponds to the eigenfunction with a 
zero-valued lower (upper) component at the origin $x=0$ (Theorem~\ref{thm:A}). The eigenvalue $\lambda_{+}$ approaches 
$\inf\ess(A)=-\Omega/2$ at $\Omega=\eta^{2}+\gamma^{2}/4$ and then disappears (for details, refer to Remark~\ref{rem:Fig}). 
Resonant states of $A$ are drawn by the yellow curve (R).
}
\label{fig:EigenvaluesA}
\end{figure}

\begin{proof}[Proof of Theorem~\ref{thm:A}]
First off, we note that, for $\lambda\in\disc(A)$, $\lambda\neq\lambda_{0}$ due to Lemma~\ref{lem:ess}-(i). Next, combining
(\ref{eq:fp}) with Lemma~\ref{lem:ess}-(i) we immediately infer (see also the proof of Lemma~\ref{lem:ess}-(i) and in particular
(\ref{eq:Green2})) item (2) of the theorem. But then, it holds $f(0_{+})=f(0_{-})\equiv f(0)$. By solving
$(I\ot  I+\gamma(A^{0}-\lambda\ot  I)^{-1}(0))f(0)=0$, we recover $\disc(B)\backslash\so(A)$ ($\so(B)=\so(A)$ by 
Theorem~\ref{thm:EA}-(viii)).

In order to accomplish the proof of (1), it therefore remains to establish $\mrm{ker}(\lambda\ot  I-A)$ 
($\lambda\in\disc(A)\backslash\so(A)$) thus proving that items (3a)--(3b) yield $\disc(A)=\disc(B)$, which in turn is found by
computing the poles of $R_{z}(A)$ in Lemma~\ref{lem:ess}-(i).

We solve the characteristic equation for $H_{0}f=\lambda f$; see (\ref{eq:kss2}). Then

\begin{align}
f(x)=&\left(\begin{matrix}c_{1} \\ c_{3} \end{matrix}\right)e^{k_{1}x}+\left(\begin{matrix}c_{2} \\ c_{4} \end{matrix}\right)e^{k_{2}x}
\quad(x>0;c_{1},\ldots,c_{4}\in\C;\mrm{Re}\:k_{j}<0;j=1,2), \nonumber \\
=&\left(\begin{matrix}\tilde{c}_{1} \\ \tilde{c}_{3} \end{matrix}\right)e^{-k_{1}x}+
\left(\begin{matrix}\tilde{c}_{2} \\ \tilde{c}_{4} \end{matrix}\right)e^{-k_{2}x}
\quad(x<0;\tilde{c}_{1},\ldots,\tilde{c}_{4}\in\C;\mrm{Re}\:k_{j}<0;j=1,2), \label{eq:fxstart}
\end{align}

\noindent{}where

\begin{equation}
k_{ss^{\prime}}=s^{\prime}\sqrt{-\lambda-\eta^{2}/2+is\eta\sqrt{\lambda_{0}-\lambda} }\quad
(\lambda_{0}=-(\eta^{2}+(\Omega/\eta)^{2})/4)
\label{eq:kssp}
\end{equation}

\noindent{}($k_{1}\equiv k_{+-};k_{2}\equiv k_{-,s^{\prime}};s,s^{\prime}=\pm1;\eta>0$). The condition $\mrm{Re}\:k_{j}<0$ ($j=1,2$) is 
due to $f\in D(A)$ (recall (\ref{eq:A})). The boundary condition in $D(A)$, provided $f(0_{+})=f(0_{-})$, yields

\begin{equation}
\left(\begin{matrix}c_{1}+c_{2} \\ c_{3}+c_{4} \end{matrix}\right)=
\left(\begin{matrix}\tilde{c}_{1}+\tilde{c}_{2} \\ \tilde{c}_{3}+\tilde{c}_{4} \end{matrix}\right),\quad
\gamma
\left(\begin{matrix}c_{1}+c_{2} \\ c_{3}+c_{4} \end{matrix}\right)=
\left(\begin{matrix}k_{1}(c_{1}+\tilde{c}_{1})+k_{2}(c_{2}+\tilde{c}_{2}) \\ 
k_{1}(c_{3}+\tilde{c}_{3})+k_{2}(c_{4}+\tilde{c}_{4}) \end{matrix}\right).
\label{eq:F1}
\end{equation}

\noindent{}We now substitute obtained functions $f$ in $H_{0}f=\lambda f$ and find that

\begin{align}
&
c_{1}(k_{1}^{2}+\lambda-\Omega/2)+c_{3}\eta k_{1}=0,\quad
c_{2}(k_{2}^{2}+\lambda-\Omega/2)+c_{4}\eta k_{2}=0, \nonumber \\
&
c_{3}(k_{1}^{2}+\lambda+\Omega/2)-c_{1}\eta k_{1}=0,\quad
c_{4}(k_{2}^{2}+\lambda+\Omega/2)-c_{2}\eta k_{2}=0, \nonumber \\
&
\tilde{c}_{1}(k_{1}^{2}+\lambda-\Omega/2)-\tilde{c}_{3}\eta k_{1}=0,\quad
\tilde{c}_{2}(k_{2}^{2}+\lambda-\Omega/2)-\tilde{c}_{4}\eta k_{2}=0, \nonumber \\
&
\tilde{c}_{3}(k_{1}^{2}+\lambda+\Omega/2)+\tilde{c}_{1}\eta k_{1}=0,\quad
\tilde{c}_{4}(k_{2}^{2}+\lambda+\Omega/2)+\tilde{c}_{2}\eta k_{2}=0. \label{eq:F2}
\end{align}

\noindent{}We need to solve the system of equations (\ref{eq:F1})--(\ref{eq:F2}). In particular, one finds from (\ref{eq:F2}),

\begin{equation}
c_{3}=c_{1}Y_{1}^{(1)},\quad c_{4}=c_{2}Y_{2}^{(2)},\quad \tilde{c}_{3}=\tilde{c}_{1}Y_{3}^{(1)},\quad
\tilde{c}_{4}=\tilde{c}_{2}Y_{4}^{(2)},
\label{eq:c34}
\end{equation}

\noindent{}where

\begin{equation}
Y_{j}^{(s)}=\frac{a_{j}\Omega+b_{j}\sqrt{\Omega^{2}-(2\eta k_{s})^{2}} }{2\eta k_{s}}\quad
(j=1,\ldots,4; s=1,2),
\label{eq:Yjs}
\end{equation}

\noindent{}and $a_{1}=a_{2}=+1$, $a_{3}=a_{4}=-1$, $b_{j}=\pm 1$ for all $j=1,\ldots,4$. Hence $Y_{j}^{(s)}=-ib_{j}$ for
$\Omega=0$.

For example, let $j=1$, $s=1$. From the first and third equations in (\ref{eq:F2}) one gets that

\begin{align*}
&\left\{\begin{array}{l}
c_{1}(k_{1}^{2}+\lambda-\Omega/2)+c_{3}\eta k_{1}=0, \\ \\
c_{3}(k_{1}^{2}+\lambda+\Omega/2)-c_{1}\eta k_{1}=0
\end{array}\right.\Longrightarrow 
\left\{\begin{array}{l}
c_{1}c_{3}(k_{1}^{2}+\lambda-\Omega/2)+c_{3}^{2}\eta k_{1}=0, \\ \\
c_{1}c_{3}(k_{1}^{2}+\lambda+\Omega/2)-c_{1}^{2}\eta k_{1}=0
\end{array}\right. \\
&\Longrightarrow
c_{1}c_{3}\Omega=\eta k_{1}(c_{1}^{2}+c_{3}^{2})\Longrightarrow c_{3}=c_{1}Y_{1}^{(1)},
\end{align*}

\noindent{}and similarly for the remaining $j=2,3,4$.

By (\ref{eq:c34})--(\ref{eq:Yjs}), there are $2^{4}=16$ possible solutions with respect to $a_{j}$ and $b_{j}$ for $j=1,\ldots,4$. These 
are tabulated in Tab.~\ref{tab:1}.

The number of distributions in Tab.~\ref{tab:1} must be reduced with the help of (\ref{eq:F1}). By (\ref{eq:F1}), one can express
$\tilde{c}_{j}$ in terms of $c_{j}$ ($j=1,\ldots,4$). Namely,

\begin{table}[htp!]
\centering
\caption{All possible solutions of (\ref{eq:F2}) with respect to $\{c_{3},c_{4},\tilde{c}_{3},\tilde{c}_{4}\}$ 
for $a_{j}$, $b_{j}=\pm 1$ for $j=1,\ldots,4$ given in (\ref{eq:c34})--(\ref{eq:Yjs}).}
\begin{tabular}{|rccccccccrrrr|}\hline
$N$  & $a_{1}$ & $b_{1}$ & $a_{2}$ & $b_{2}$ & $a_{3}$ & $b_{3}$ & $a_{4}$ & $b_{4}$ & $b_{1}-b_{3}$ & $b_{2}-b_{3}$ &
$b_{1}-b_{4}$ & $b_{2}-b_{4}$ \\ \hline\hline
$1$  & $+$ & $-$ & $+$ & $-$ & $-$ & $-$ & $-$ & $-$ & $0$ & $0$ & $0$ & $0$ \\
$2$  & $+$ & $-$ & $+$ & $-$ & $-$ & $-$ & $-$ & $+$ & $0$ & $0$ & $-2$ & $-2$ \\
$3$  & $+$ & $-$ & $+$ & $-$ & $-$ & $+$ & $-$ & $-$ & $-2$ & $-2$ & $0$ & $0$ \\
$4$  & $+$ & $-$ & $+$ & $-$ & $-$ & $+$ & $-$ & $+$ & $-2$ & $-2$ & $-2$ & $-2$ \\
$5$  & $+$ & $-$ & $+$ & $+$ & $-$ & $-$ & $-$ & $-$ & $0$ & $+2$ & $0$ & $+2$ \\
$6$  & $+$ & $-$ & $+$ & $+$ & $-$ & $-$ & $-$ & $+$ & $0$ & $+2$ & $-2$ & $0$ \\
$7$  & $+$ & $-$ & $+$ & $+$ & $-$ & $+$ & $-$ & $-$ & $-2$ & $0$ & $0$ & $+2$ \\
$8$  & $+$ & $-$ & $+$ & $+$ & $-$ & $+$ & $-$ & $+$ & $-2$ & $0$ & $-2$ & $0$ \\
$9$  & $+$ & $+$ & $+$ & $-$ & $-$ & $-$ & $-$ & $-$ & $+2$ & $0$ & $+2$ & $0$ \\
$10$ & $+$ & $+$ & $+$ & $-$ & $-$ & $-$ & $-$ & $+$ & $+2$ & $0$ & $0$ & $-2$ \\
$11$ & $+$ & $+$ & $+$ & $-$ & $-$ & $+$ & $-$ & $-$ & $0$ & $-2$ & $+2$ & $0$ \\
$12$ & $+$ & $+$ & $+$ & $-$ & $-$ & $+$ & $-$ & $+$ & $0$ & $-2$ & $0$ & $-2$ \\
$13$ & $+$ & $+$ & $+$ & $+$ & $-$ & $-$ & $-$ & $-$ & $+2$ & $+2$ & $+2$ & $+2$ \\
$14$ & $+$ & $+$ & $+$ & $+$ & $-$ & $-$ & $-$ & $+$ & $+2$ & $+2$ & $0$ & $0$ \\
$15$ & $+$ & $+$ & $+$ & $+$ & $-$ & $+$ & $-$ & $-$ & $0$ & $0$ & $+2$ & $+2$ \\
$16$ & $+$ & $+$ & $+$ & $+$ & $-$ & $+$ & $-$ & $+$ & $0$ & $0$ & $0$ & $0$ \\
\hline
\end{tabular}
\label{tab:1}
\end{table}

\begin{subequations}\label{eq:further}
\begin{align}
&\tilde{c}_{1}(k_{1}-k_{2})=c_{1}(\gamma-k_{1}-k_{2})+c_{2}(\gamma-2k_{2}), \nonumber \\
&\tilde{c}_{2}(k_{1}-k_{2})=c_{1}(2k_{1}-\gamma)+c_{2}(k_{1}+k_{2}-\gamma), \label{eq:further-a} \\
\intertext{and}
&\tilde{c}_{3}(k_{1}-k_{2})=c_{3}(\gamma-k_{1}-k_{2})+c_{4}(\gamma-2k_{2}), \nonumber \\
&\tilde{c}_{4}(k_{1}-k_{2})=c_{3}(2k_{1}-\gamma)+c_{4}(k_{1}+k_{2}-\gamma). \label{eq:further-b}
\end{align}
\end{subequations}

\noindent{}By (\ref{eq:c34}), substitute $\tilde{c}_{3}$, $c_{3}$ and $c_{4}$ in the first equation of (\ref{eq:further-b}) and get

$$
\tilde{c}_{1}Y_{3}^{(1)}(k_{1}-k_{2})=c_{1}Y_{1}^{(1)}(\gamma-k_{1}-k_{2})+c_{2}Y_{2}^{(2)}(\gamma-2k_{2}).
$$

\noindent{}Now multiply the first equation of (\ref{eq:further-a}) by $Y_{3}^{(1)}$ and subtract both obtained equations so that
$\tilde{c}_{1}$ is eliminated,

\begin{subequations}\label{eq:c12Y}
\begin{equation}
0=c_{1}(\gamma-k_{1}-k_{2})\bigl(Y_{1}^{(1)}-Y_{3}^{(1)}\bigr)+c_{2}(\gamma-2k_{2})\bigl(Y_{2}^{(2)}-Y_{3}^{(1)}\bigr).
\label{eq:c12Y-a}
\end{equation}

\noindent{}Similarly, by using (\ref{eq:c34}), substitute $\tilde{c}_{4}$, $c_{3}$ and $c_{4}$ in the second equation of 
(\ref{eq:further-b}) and get

$$
\tilde{c}_{2}Y_{4}^{(2)}(k_{1}-k_{2})=c_{1}Y_{1}^{(1)}(2k_{1}-\gamma)+c_{2}Y_{2}^{(2)}(k_{1}+k_{2}-\gamma).
$$

\noindent{}Multiply the second equation of (\ref{eq:further-a}) by $Y_{4}^{(2)}$ and subtract both obtained equations so that
$\tilde{c}_{2}$ is eliminated,

\begin{equation}
0=c_{1}(2k_{1}-\gamma)\bigl(Y_{1}^{(1)}-Y_{4}^{(2)}\bigr)+c_{2}(k_{1}+k_{2}-\gamma)\bigl(Y_{2}^{(2)}-Y_{4}^{(2)}\bigr).
\label{eq:c12Y-b}
\end{equation}
\end{subequations}

\noindent{}By using (\ref{eq:Yjs}), equations (\ref{eq:c12Y}) can be rewritten explicitly as follows

\begin{align*}
0=& c_{1}k_{2}(\gamma-k_{1}-k_{2})\bigl(2\Omega+(b_{1}-b_{3})\bigl(\Omega^{2}-(2\eta k_{1})^{2}\bigr)^{\frac{1}{2}} \bigr)  \\
&+c_{2}(\gamma-2k_{2})\bigl(\Omega(k_{1}+k_{2})+b_{2}k_{1}\bigl(\Omega^{2}-(2\eta k_{2})^{2}\bigr)^{\frac{1}{2}}  \\
&-b_{3}k_{2}\bigl(\Omega^{2}-(2\eta k_{1})^{2}\bigr)^{\frac{1}{2}} \bigr),  \\
\intertext{and}
0=& c_{1}(2k_{1}-\gamma)\bigl(\Omega(k_{1}+k_{2})+b_{1}k_{2}\bigl(\Omega^{2}-(2\eta k_{1})^{2}\bigr)^{\frac{1}{2}}  \\
&-b_{4}k_{1}\bigl(\Omega^{2}-(2\eta k_{2})^{2}\bigr)^{\frac{1}{2}} \bigr)+c_{2}k_{1}(k_{1}+k_{2}-\gamma)\bigl(2\Omega  \\
&+(b_{2}-b_{4})\bigl(\Omega^{2}-(2\eta k_{2})^{2}\bigr)^{\frac{1}{2}} \bigr). 
\end{align*}

\noindent{}By noting that $c_{1}$ and $c_{2}$ are two independent constants, we can subtract both equations and separate the expressions
at $c_{1}$ and $c_{2}$ one from another. Then

$$
E_{\Omega}(k_{1},k_{2})\equiv 0,\quad \varphi E_{\Omega}(k_{1},k_{2})\equiv 0,
$$

\noindent{}where

\begin{align*}
E_{\Omega}(k_{1},k_{2})=&\Omega[\gamma(k_{1}+3k_{2})-2(k_{1}+k_{2})^{2}]+b_{4}k_{1}(2k_{1}-\gamma)
[\Omega^{2}-(2\eta k_{2})^{2}]^{\frac{1}{2}} \\
&+k_{2}[b_{3}(k_{1}+k_{2}-\gamma)-b_{1}(3k_{1}+k_{2}-2\gamma)][\Omega^{2}-(2\eta k_{1})^{2}]^{\frac{1}{2}},
\end{align*}

\noindent{}with a one-to-one map $\varphi\co k_{1}\mapsto k_{2}$, $k_{2}\mapsto k_{1}$,
$b_{1}\mapsto b_{2}$, $b_{2}\mapsto b_{1}$, $b_{3}\mapsto b_{4}$, and $b_{4}\mapsto b_{3}$.
Then $\varphi^{n}=I$ (identity) for $n=0,2,4,\ldots$, and $\varphi^{n}=\varphi$ for 
$n=1,3,5,\ldots$ Equation $E_{0}\equiv0$ holds for the distributions (Tab.~\ref{tab:1}) numbered by $N=2$, $4$, $6$, $8$ and $9$, 
$11$, $13$, $15$. On the other hand, $E_{\Omega}$ with $\Omega>0$ is well defined for $N=2$, $6$ and $11$, $15$. Therefore, we deduce that
for $\Omega\geq0$, $E_{\Omega}$ makes sense if $N=2$, $6$ and $11$, $15$.

Expression $E_{\Omega}$ can be represented by the sum of $F_{\Omega}$ and $G_{\Omega}$, where both $F_{\Omega}$ and $G_{\Omega}$ are 
invariant under the action of $\varphi$, namely,

$$
F_{\Omega}(k_{1},k_{2})=\Omega(k_{1}+k_{2})[\gamma-2(k_{1}+k_{2})],\quad \varphi F_{\Omega}(k_{1},k_{2})=F_{\Omega}(k_{1},k_{2}),
$$

\noindent{}and $G_{\Omega}$ is defined by

\begin{align*}
G_{\Omega}(k_{1},k_{2})=&
2\gamma\Omega k_{2}+b_{4}k_{1}(2k_{1}-\gamma)[\Omega^{2}-(2\eta k_{2})^{2}]^{\frac{1}{2}} \\
&+k_{2}[b_{3}(k_{1}+k_{2}-\gamma)-b_{1}(3k_{1}+k_{2}-2\gamma)][\Omega^{2}-(2\eta k_{1})^{2}]^{\frac{1}{2}}.
\end{align*}

\noindent{}Then $G_{\Omega}$ satisfies

$$
G_{\Omega}(k_{1},k_{2})=\varphi G_{\Omega}(k_{1},k_{2})=-F_{\Omega}(k_{1},k_{2})\quad(\text{since } E_{\Omega}\equiv0)
$$

\noindent{}and

\begin{align*}
\varphi^{n}G_{\Omega}(k_{1},k_{2})=        G_{\Omega}(k_{1},k_{2})\quad\text{for}&\quad n=0,2,4,\ldots, \\
\varphi^{n}G_{\Omega}(k_{1},k_{2})=\varphi G_{\Omega}(k_{1},k_{2})\quad\text{for}&\quad n=1,3,5,\ldots
\end{align*}

\noindent{}Then $(\varphi-I)G_{\Omega}=0$ yields

\begin{align}
(\varphi-I)G_{\Omega}(k_{1},k_{2})=& 2\gamma\Omega(k_{1}-k_{2})+k_{2}[b_{1}(3k_{1}+k_{2}-2\gamma)-b_{3}(k_{1}-k_{2})] \nonumber \\
&\times [\Omega^{2}-(2\eta k_{1})^{2}]^{\frac{1}{2}}-k_{1}[b_{2}(3k_{2}+k_{1}-2\gamma) \nonumber \\
&+b_{4}(k_{1}-k_{2})][\Omega^{2}-(2\eta k_{2})^{2}]^{\frac{1}{2}}=0. \label{eq:solvek2}
\end{align}

\noindent{}Equation (\ref{eq:solvek2}) shows that, depending on $16$ distributions in Tab.~\ref{tab:1}, four distinct classes can be 
considered.

\begin{subequations}\label{eq:distr4}
\begin{align}
&(I):\:E_{\Omega}^{(1)}(k_{1},k_{2})\equiv 0,\;\text{with}\;\; E_{\Omega}^{(1)}(k_{1},k_{2})
=\gamma\Omega(k_{1}-k_{2}) \nonumber \\
&+(k_{1}+k_{2}-\gamma)\Bigl(b_{1}k_{2}[\Omega^{2}-(2\eta k_{1})^{2}]^{\frac{1}{2}}
-b_{2}k_{1}[\Omega^{2}-(2\eta k_{2})^{2}]^{\frac{1}{2}}\Bigr) \label{eq:distr4-a} \\
\intertext{($b_{1}=b_{3}$, $b_{2}=b_{4}$),}
&(II):\:E_{\Omega}^{(2)}(k_{1},k_{2})\equiv 0,\;\text{with}\;\; E_{\Omega}^{(2)}(k_{1},k_{2})
=\gamma\Omega(k_{1}-k_{2}) \nonumber \\
&+b_{1}k_{2}(k_{1}+k_{2}-\gamma)[\Omega^{2}-(2\eta k_{1})^{2}]^{\frac{1}{2}}
-b_{2}k_{1}(2k_{2}-\gamma)[\Omega^{2}-(2\eta k_{2})^{2}]^{\frac{1}{2}} \label{eq:distr4-b} \\
\intertext{($b_{1}=b_{3}$, $b_{2}=-b_{4}$),}
&(III):\:E_{\Omega}^{(3)}(k_{1},k_{2})\equiv 0,\;\text{with}\;\; E_{\Omega}^{(3)}(k_{1},k_{2})=-\varphi_{1}E_{\Omega}^{(2)}(k_{1},k_{2})
\label{eq:distr4-c} \\
\intertext{($b_{1}=-b_{3}$, $b_{2}=b_{4}$ and $\varphi_{1}\co k_{1}\mapsto k_{2}$, $k_{2}\mapsto k_{1}$,
$b_{1}\mapsto b_{2}$, $b_{2}\mapsto b_{1}$),} 
&(IV):\:E_{\Omega}^{(4)}(k_{1},k_{2})\equiv 0,\;\text{with}\;\; E_{\Omega}^{(4)}(k_{1},k_{2})
=\gamma\Omega(k_{1}-k_{2}) \nonumber \\
&+b_{1}k_{2}(2k_{1}-\gamma)[\Omega^{2}-(2\eta k_{1})^{2}]^{\frac{1}{2}}
-b_{2}k_{1}(2k_{2}-\gamma)[\Omega^{2}-(2\eta k_{2})^{2}]^{\frac{1}{2}} \label{eq:distr4-d}
\end{align}
\end{subequations}

\noindent{}($b_{1}=-b_{3}$, $b_{2}=-b_{4}$).

By the isomorphism in (\ref{eq:distr4-c}), it suffices to consider three classes: $(I)$, $(II)$, 
$(IV)$.

Class $(I)$. Given $\Omega>0$, the equation $E_{\Omega}^{(1)}\equiv0$ (\ref{eq:distr4-a}) holds for the distributions numbered by $N=1$, 
$6$, $11$ and $16$. If, however, $\Omega=0$, then $E_{0}^{(1)}\equiv0$ holds for all $k_{1}$, $k_{2}$, which is inconsistent with the 
point spectrum of $A$. Subsequently, class $(I)$ is improper.

Class $(II)$. For $\Omega>0$, $E_{\Omega}^{(2)}\equiv0$ (\ref{eq:distr4-b}) holds for the distributions numbered by $N=2$, $5$, $12$ and 
$15$. Due to the isomorphism $\varphi_{1}$, the number of distributions decreases to $N=2$, $3$, $5$, $8$, $9$, $12$, 
$14$ and $15$. But $E_{0}^{(2)}\equiv0$ yields $k_{1}(k_{1}+k_{2}-2\gamma)+k_{2}(k_{1}-3k_{2}+2\gamma)=0$ which is satisfied only
for $k_{1}=k_{2}=\gamma/2$, hence improper due to $\lambda\neq\lambda_{0}$.

Class $(IV)$. For $\Omega>0$, $E_{\Omega}^{(4)}\equiv0$ (\ref{eq:distr4-d}) holds for the distributions numbered by $N=4$, $7$, $10$ and 
$13$. For $\Omega=0$, $E_{0}^{(4)}\equiv0$ yields a correct relation $k_{1}+k_{2}=\gamma$. Possible distributions are numbered by
$N=7$ and $N=10$. 

As a result, we have found that $E_{\Omega}^{(4)}\equiv0$ is the only one correct equation which holds for all $\Omega\geq0$.
The associated distributions in Tab.~\ref{tab:1} are numbered by $N=7$ and $10$.

By solving (\ref{eq:distr4-d}), we find that

\begin{equation}
k_{1}+k_{2}=\gamma(1+\chi_{\Omega}),
\label{eq:maineq}
\end{equation}

\noindent{}where

\begin{equation}
\chi_{\Omega}=\Omega\cdot\frac{-\gamma^{2}\Omega+2k_{1}k_{2}\bigl(\Omega\pm
[\Omega^{2}-(\gamma\eta)^{2}+(2\eta)^{2}k_{1}k_{2}]^{\frac{1}{2}}\bigr) }{2[(2\eta k_{1}k_{2})^{2}+(\gamma\Omega)^{2}]}
\label{eq:maineq2}
\end{equation}

\noindent{}($\Omega,\eta\geq0$),
$\chi_{0}=0$ and $\gamma<0$. As it should be by (\ref{eq:solvek2}), equation (\ref{eq:maineq}) is invariant under the 
action of $\varphi$ as well as $\varphi_{1}$. 

Recalling that $k_{1}k_{2}=s^{\prime}[\lambda^{2}-(\Omega/2)^{2}]^{1/2}$ ($s^{\prime}=\pm1$), one can construct the equation for the 
eigenvalues $\lambda$. By (\ref{eq:maineq}), $\lambda$ satisfies the following cubic equation

\begin{align}
& (8\eta)^{2}\lambda^{3}+16[\eta^{2}(\gamma^{2}+\eta^{2})+\Omega(\Omega\pm 4\eta^{2})]\lambda^{2} \nonumber \\
&\pm 8\Omega[2\Omega^{2}+(\gamma^{2}+2\eta^{2})(\eta^{2}\pm\Omega)]\lambda+
\Omega^{2}[4\eta^{4}+(\gamma^{2}\pm 2\Omega)^{2}]=0 \label{eq:lambdaobt}
\end{align}

\noindent{}($\Omega\geq0$), provided $\mrm{Re}\:k_{j}<0$ for $j=1,2$. Note that the sign $\pm$ corresponds to that in
(\ref{eq:maineq2}).

Now, it is necessary to show that the eigenvalues $\lambda$, which satisfy (\ref{eq:lambdaobt}), are also in $\disc(B)\backslash\so(A)$, 
thus accomplishing the proof of Theorem~\ref{thm:A}-(1), and that the eigenfunctions of $\disc(A)\backslash\so(A)$ are as in 
(\ref{eq:AAAA})--(\ref{eq:AAAA1}), thus giving Theorem~\ref{thm:A}-(3a) and (3b). 

We solve (\ref{eq:F2}) with respect to $c_{3}$, $c_{4}$ and $\tilde{c}_{3}$, $\tilde{c}_{4}$, by assuming that $\eta>0$,

\begin{align*}
c_{3}=&c_{1}\:\frac{\eta k_{1}}{k_{1}^{2}+\lambda+\Omega/2}=-c_{1}\:\frac{k_{1}^{2}+\lambda-\Omega/2}{\eta k_{1}}, \\
c_{4}=&c_{2}\:\frac{\eta k_{2}}{k_{2}^{2}+\lambda+\Omega/2}=-c_{2}\:\frac{k_{2}^{2}+\lambda-\Omega/2}{\eta k_{2}}, \\
\tilde{c}_{3}=&-\tilde{c}_{1}\:\frac{\eta k_{1}}{k_{1}^{2}+\lambda+\Omega/2}=
\tilde{c}_{1}\:\frac{k_{1}^{2}+\lambda-\Omega/2}{\eta k_{1}}, \\
 \tilde{c}_{4}=&-\tilde{c}_{2}\:\frac{\eta k_{2}}{k_{2}^{2}+\lambda+\Omega/2}=
\tilde{c}_{2}\:\frac{k_{2}^{2}+\lambda-\Omega/2}{\eta k_{2}}.
\end{align*}

\noindent{}We note that each equality in every row can be chosen arbitrarily; we choose the first one. Substitute obtained 
expressions in (\ref{eq:fxstart}) and find by (\ref{eq:F1}),

\begin{align*}
&f(0_{+})=c_{1}\left(\begin{matrix}1 \\ \frac{\eta k_{1}}{k_{1}^{2}+\lambda+\Omega/2} \end{matrix}\right)+
c_{2}\left(\begin{matrix}1 \\ \frac{\eta k_{2}}{k_{2}^{2}+\lambda+\Omega/2} \end{matrix}\right),\;\;
f(0_{-})=\tilde{c}_{1}\left(\begin{matrix}1 \\ \frac{-\eta k_{1}}{k_{1}^{2}+\lambda+\Omega/2} \end{matrix}\right)+
\tilde{c}_{2}\left(\begin{matrix}1 \\ \frac{-\eta k_{2}}{k_{2}^{2}+\lambda+\Omega/2} \end{matrix}\right), \\
&f^{\prime}(0_{+})=c_{1}\left(\begin{matrix}k_{1} \\ \frac{\eta k_{1}^{2}}{k_{1}^{2}+\lambda+\Omega/2} \end{matrix}\right)+
c_{2}\left(\begin{matrix}k_{2} \\ \frac{\eta k_{2}^{2}}{k_{2}^{2}+\lambda+\Omega/2} \end{matrix}\right),\;\;
f^{\prime}(0_{-})=\tilde{c}_{1}\left(\begin{matrix}-k_{1} \\ \frac{\eta k_{1}^{2}}{k_{1}^{2}+\lambda+\Omega/2} \end{matrix}\right)+
\tilde{c}_{2}\left(\begin{matrix}-k_{2} \\ \frac{\eta k_{2}^{2}}{k_{2}^{2}+\lambda+\Omega/2} \end{matrix}\right).
\end{align*}

\noindent{}These functions, with $f(0_{+})=f(0_{-})$, are in $D(A)$. Hence the boundary condition given by
$\gamma(f(0_{+})+f(0_{-}))/2=f^{\prime}(0_{+})-f^{\prime}(0_{-})$ yields

\begin{subequations}\label{eq:b12}
\begin{align}
0=&(c_{1}+c_{2})\left(\gamma-\frac{2k_{1}k_{2}(k_{1}+k_{2})}{k_{1}k_{2}-\lambda-\Omega/2}\right),  \\
0=&\eta(c_{1}-\tilde{c}_{1})
\left(\frac{k_{1}(2k_{1}-\gamma)}{k_{1}^{2}+\lambda+\Omega/2}-\frac{k_{2}(2k_{2}-\gamma)}{k_{2}^{2}+\lambda+\Omega/2}\right). 
\end{align}
\end{subequations}

\noindent{}By (\ref{eq:b12}), four possible cases are then considered, provided $\eta>0$:

Case (1). $c_{1}+c_{2}=0$ and $c_{1}-\tilde{c}_{1}=0$. By (\ref{eq:fxstart}), $c_{1}+c_{2}=\tilde{c}_{1}+\tilde{c}_{2}=0$. Hence
$\tilde{c}_{2}=-\tilde{c}_{1}=-c_{1}$. By (\ref{eq:F1}), $\tilde{c}_{2}(k_{1}-k_{2})=c_{1}(2k_{1}-\gamma)+
c_{2}(k_{1}+k_{2}-\gamma)$ (see also (\ref{eq:further-a})). Hence $c_{1}(k_{1}-k_{2})=0$. If $c_{1}=0$, then $f\equiv0$, hence trivial.
If $k_{1}=k_{2}$, then $\lambda=\lambda_{0}$, by (\ref{eq:kssp}), and $f\equiv0$, by (\ref{eq:fxstart}); hence improper again.

Case (2). 

$$
c_{1}+c_{2}=0\quad\text{and}\quad 
\frac{k_{1}(2k_{1}-\gamma)}{k_{1}^{2}+\lambda+\Omega/2}-\frac{k_{2}(2k_{2}-\gamma)}{k_{2}^{2}+\lambda+\Omega/2}=0
\Longrightarrow
\gamma=-\frac{2(k_{1}+k_{2})(\lambda+\Omega/2)}{k_{1}k_{2}-\lambda-\Omega/2}.
$$

\noindent{}If we expand the latter equation by using (\ref{eq:kssp}), this agrees with (\ref{eq:lambdaobt}) for the upper sign.
By noting that $k_{j}=ip_{j}$ for $j=1,2$, and $p_{j}$ as in the theorem, we find that the correspondence is one-to-one with
the eigenvalues in $\disc(B)\backslash\so(A)$ obtained by setting the lower sign.

By (\ref{eq:fxstart}), $c_{1}+c_{2}=\tilde{c}_{1}+\tilde{c}_{2}=0$, and thus $\tilde{c}_{2}=-\tilde{c}_{1}$. Then 
(\ref{eq:further-a}) yields $\tilde{c}_{1}=-c_{1}$ and $\tilde{c}_{2}=c_{1}$. The substitution of these coefficients in
(\ref{eq:fxstart}) gives (\ref{eq:AAAA1}), with $k_{j}=ip_{j}$ ($j=1,2$), $C\equiv c_{1}\in\C$.

Case (3).

$$
\gamma-\frac{2k_{1}k_{2}(k_{1}+k_{2})}{k_{1}k_{2}-\lambda-\Omega/2}=0\quad\text{and}\quad c_{1}-\tilde{c}_{1}=0.
$$

\noindent{}
Similarly to the previous case, by expanding the former equation with the help of (\ref{eq:kssp}), we establish (\ref{eq:lambdaobt})
with the lower sign. Subsequently, this corresponds to the upper sign in $\disc(B)\backslash\so(A)$.

The latter equation, $c_{1}-\tilde{c}_{1}=0$, along with (\ref{eq:fxstart}) yields

$$
\tilde{c}_{1}=c_{1},\quad\tilde{c}_{2}=c_{2}=-c_{1}\:\frac{k_{1}}{k_{2}}\cdot\frac{k_{2}^{2}+\lambda+\Omega/2}{
k_{1}^{2}+\lambda+\Omega/2 }.
$$

\noindent{}Substitute obtained coefficients in (\ref{eq:fxstart}) and get (\ref{eq:AAAA}), with $k_{j}=ip_{j}$ ($j=1,2$) and the
coefficient $C\equiv c_{1}p_{1}/(\lambda+\Omega/2-p_{1}^{2})\in\C$ (note that the denominator is nonzero unless $\lambda$ is in the 
essential spectrum).

Case (4).

$$
\gamma=\frac{2k_{1}k_{2}(k_{1}+k_{2})}{k_{1}k_{2}-\lambda-\Omega/2}\quad\text{and}\quad
\gamma=-\frac{2(k_{1}+k_{2})(\lambda+\Omega/2)}{k_{1}k_{2}-\lambda-\Omega/2}.
$$

\noindent{}The combination of both equations yields $(k_{1}+k_{2})(k_{1}k_{2}+\lambda+\Omega/2)=0$. If 
$k_{1}k_{2}+\lambda+\Omega/2=0$, then, recalling that (refer to (\ref{eq:kssp})) $k_{1}k_{2}=s^{\prime}\sqrt{\lambda^{2}-(\Omega/2)^{2}}$
($s^{\prime}=\pm1$), it holds $\lambda=-\Omega/2$, hence improper. If, however, $k_{1}+k_{2}=0$, then $\lambda=\lambda_{0}$, by 
(\ref{eq:kssp}), hence improper again.

As a result, Cases (2)--(3) accomplish the proof of items (1) and (3a)--(3b) of Theorem~\ref{thm:A}.

We now concentrate on (3c). For $\eta=0$, equation $H_{0}f=\lambda f$, $f\in D(A)$, is easy to deal with since the components $f_{1}$ and
$f_{2}$ are separated and thus can be solved independently one from another: $f_{1}^{\prime\prime}+(\lambda-\Omega/2)f_{1}=0$,
$f_{2}^{\prime\prime}+(\lambda+\Omega/2)f_{2}=0$. By substituting obtained exponents in the boundary condition we get (3c). Moreover,
the condition $\gamma<-2\sqrt{\Omega}$ is obtained from the inspection of the resolvent in Lemma~\ref{lem:ess}-(i), where one requires
$\mrm{Im}\:p_{j}>0$ for $j=1,2$. For $\eta=0$, $z<-\Omega/2$, and hence $-\gamma^{2}/4+\Omega/2<-\Omega/2$ thus yielding
$\gamma<-2\sqrt{\Omega}$. Otherwise, only one eigenvalue $-\gamma^{2}/4-\Omega/2$ remains.

In particular, this also proves that $(\disc(A)\backslash\so(A))\cap\ess(A)=\varnothing$ (see item (4) of the theorem) for $\eta=0$, since 
$J(0,\Omega)=-\Omega/2$. For arbitrary spin-orbit coupling $\eta>0$, let us examine the conditions $\mrm{Im}\:p_{j}>0$ for $j=1,2$.
It suffices to show the converse for at least one $p_{j}$.

Let $j=1$ and $0<\Omega\leq\eta^{2}$. Then $J(\eta,\Omega)=\lambda_{0}$. Assume that the eigenvalue 
$\lambda=\lambda_{0}+\nu$ for some real $\nu>0$. Then it holds $p_{1}=\sqrt{\lambda_{0}+\nu+\eta^{2}/2+\eta\sqrt{\nu}}$. 
But $\lambda_{0}+\eta^{2}/2=(\eta^{4}-\Omega^{2})/(4\eta^{2})\geq0$ for all $0<\Omega\leq\eta^{2}$. Hence $\mrm{Im}\:p_{1}=0$, which is 
invalid.

Let $j=1$ and $\Omega>\eta^{2}>0$. Then $J(\eta,\Omega)=-\Omega/2$. Let $\lambda=-\Omega/2+\nu$ for some $\nu>0$.
Then we have that $p_{1}=\sqrt{-a+\nu+\sqrt{a^{2}+\eta^{2}\nu }}$, where $a=(\Omega-\eta^{2})/2>0$ for all
$\Omega>\eta^{2}>0$. As seen, $\mrm{Im}\:p_{1}=0$ for all $0<\nu\leq\Omega$. Next, let $\nu=\Omega+\mu$ for some $\mu>0$,
and substitute $\lambda=\Omega/2+\mu$ in (\ref{eq:lambdaobt}). One gets that

\begin{align*}
0=&\eta^{4}\Omega^{2}+8\eta^{2}(\Omega+\mu)(\Omega^{2}+\eta^{2}(\Omega+2\mu))+
16(\Omega+\mu)^{2}(\eta^{4}+\Omega^{2}+2\eta^{2}(\Omega+2\mu)), \\
\intertext{for the lower sign, and that}
0=&16\eta^{2}\mu^{2}(\gamma^{2}+\eta^{2}+4\mu)+8\eta^{2}\mu\Omega(\gamma^{2}+4\mu)+
\Omega^{2}(\gamma^{2}+4\mu)^{2}
\end{align*}

\noindent{}for the upper one. It is evident that the above equations do not have real solutions for all $\mu>0$ for all $\Omega,\eta>0$
($\gamma<0$), since all the terms on the right-hand side are positive, whereas the left-hand side is zero. Therefore, $\mrm{Im}\:p_{1}=0$ 
for $\nu>\Omega$ as well. Subsequently, item (4) holds, and this accomplishes the proof of the theorem.
\end{proof}

\begin{rem}\label{rem:Fig}
In Fig.~\ref{fig:EigenvaluesA}, one finds that $\lambda_{+}$ vanishes for $\Omega\geq\eta^{2}+\gamma^{2}/4$, by substituting 
$\lambda=-\Omega/2$ in $\disc(A)\backslash\so(A)$ (Theorem~\ref{thm:A}-(1)) or in (\ref{eq:lambdaobt}) and solving the obtained equation 
with respect to $\Omega$. The suffix <<+>> indicates that the eigenvalue is found from $\disc(A)\backslash\so(A)$ with the plus sign
(or from (\ref{eq:lambdaobt}) with the minus sign). We also note that the condition $\lambda<\inf\ess(A)$ alone is insufficient to
derive proper bound states; this must be implemented with the requirement $\mrm{Im}\:p_{j}>0$ for $j=1,2$ as well.
\end{rem}

\section{Summary and discussion}\label{sec:summary}

In this paper, we solved the bound state problem for the spin-orbit coupled ultracold atom in a one-dimensional short-range potential
describing the impurity scattering. The potential is assumed to be approximated by the $\delta$-interaction. As a result, two
distinct realizations of the original differential expression, $H$, were proposed. The first one, $A$, is implemented through the boundary 
condition defining the domain of the operator. The second realization, $B$, has a meaning of distribution. 
Although both representatives provide identical spectra, the eigenfunctions differ in their form: Equivalence classes of
functions of $B$ supply with insufficient information concerning the (classical) behavior of eigenfunctions.

Based on the property that $H$ contains both the spin-orbit and the Raman coupling, we showed that, for nonzero
spin-orbit and Raman coupling, the spectrum is 
implemented with some extra states, in addition to those which are found by solving the eigenvalue equation directly. Extra
states, called the spin-orbit coupling induced states, have a peculiarity that the associated eigenfunctions are discontinuous
at the origin $x=0$, and that there might be a point embedded into the essential spectrum.
By (dis)continuity we assume that, although functions are defined on any subset of $\R\backslash\{0\}$,
their left ($x=0_{-}$) and right ($x=0_{+}$) representatives either coincide (continuity) or not (discontinuity). Such states
originate from the fact that the spin-orbit Hamiltonian is not purely Dirac-like or Schr\"{o}dinger-like operator but rather their
one-dimensional mixture. It turns out that $A$ ($B$) commutes with the operator which is unitarily equivalent to the one-dimensional
Dirac operator (in Weyl's form) for the particle with spin one-half moving in the Fermi pseudopotential $V_{F}$. In turn, we showed
that $V_{F}$ is a combination of both $\delta$- and $\delta^{\prime}$-interactions, where the latter accounts for the divergent
terms occurring if dealt with discontinuous functions (one has the so-called $\delta_{p}^{\prime}$-interaction).

Finally, we established the remaining part of the discrete spectrum of $A$ ($B$) and showed that the eigenvalues under consideration
are found by solving the cubic equation. Depending on the regime of the Raman coupling, that is to say, on the strength of the
Zeeman field, one observes either two or a single point in the spectrum. The associated eigenfunctions are everywhere continuous but
with zero-valued component (either upper or lower one) at the origin.

It is worth noting that the (self-adjoint) representatives $A_{0}$ and $A$ of the atom-light coupling $U$ and the Hamiltonian $H$
could serve for a tool to recover other self-adjoint extensions thus corresponding to modified point-interactions. This could be done 
with the help of Krein's formula \cite[Eq.~(6.10)]{Krein47} (see also \cite[Appendix~A]{Albeverio05}). For that purpose one needs to 
apply the resolvents of $A_{0}$ and $A$ given in Theorem~\ref{thm:A0}-(i) and Lemma~\ref{lem:ess}-(i), respectively. Following eg
\cite{Seba86,Albeverio98}, one constructs operators on the intervals $(-\infty,0)$ and $(0,\infty)$, and finds the orthonormal bases
relevant to deficiency subspaces. So defined, the operators have d.i. (2,2). The entries of the associated unitary matrix from $U(2)$ 
group thus determine all self-adjoint extensions.

\section*{Acknowledgments}

The authors gratefully acknowledge Dr.~G.~Juzeli\={u}nas who bears much of the credit for the genesis of the present paper. It is a  
pleasure to thank Dr.~I.~Spielman for useful discussions. R.J. acknowledges Dr.~I.~Spielman for warm hospitality extended to him during 
his visit at the University of Maryland, where the part of this work has been done. The present work was supported by the Research Council 
of Lithuania (No.~VP1-3.1-\v{S}MM-01-V-02-004).

\bibliographystyle{plainnat}

\begin{thebibliography}{28}
\providecommand{\natexlab}[1]{#1}
\providecommand{\url}[1]{\texttt{#1}}
\expandafter\ifx\csname urlstyle\endcsname\relax
  \providecommand{\doi}[1]{doi: #1}\else
  \providecommand{\doi}{doi: \begingroup \urlstyle{rm}\Url}\fi

\bibitem[Adams and Fournier(2003)]{Adams03}
R.~A. Adams and J.~J. Fournier.
\newblock \emph{Sobolev {S}paces}.
\newblock Elsevier Science Ltd, Oxford, UK, 2 edition, 2003.

\bibitem[Albeverio et~al.(1998)Albeverio, D\k{a}browski, and
  Kurasov]{Albeverio98}
S.~Albeverio, L.~D\k{a}browski, and P.~Kurasov.
\newblock \emph{Lett. Math. Phys.}, 45:\penalty0 33, 1998.

\bibitem[Albeverio et~al.(2005)Albeverio, Gesztesy, Hoegh-Krohn, and
  Holden]{Albeverio05}
S.~Albeverio, F.~Gesztesy, R.~Hoegh-Krohn, and H.~Holden.
\newblock \emph{Solvable {M}odels in {Q}uantum {M}echanics}.
\newblock AMS Chelsea Publishing (Providence, Rhode Island), 2 edition, 2005.

\bibitem[Anderson et~al.(2012)Anderson, Juzeli\=unas, Galitski, and
  Spielman]{Anderson12}
Brandon~M. Anderson, Gediminas Juzeli\=unas, Victor~M. Galitski, and I.~B.
  Spielman.
\newblock \emph{Phys. Rev. Lett.}, 108:\penalty0 235301, 2012.

\bibitem[Benvegn\`{u} and D\k{a}browski(1994)]{Benvegnu94}
S.~Benvegn\`{u} and L.~D\k{a}browski.
\newblock \emph{Lett. Math. Phys.}, 30:\penalty0 159, 1994.

\bibitem[Bychkov and Rashba(1984)]{Rashba84}
Y.~A. Bychkov and E.~I. Rashba.
\newblock \emph{J. Phys. C}, 17:\penalty0 6039, 1984.

\bibitem[Campbell et~al.(2011)Campbell, Juzeli\=unas, and Spielman]{Campbell11}
D.~L. Campbell, G.~Juzeli\=unas, and I.~B. Spielman.
\newblock \emph{Phys. Rev. A}, 84:\penalty0 025602, 2011.

\bibitem[Cheuk et~al.(2012)Cheuk, Sommer, Hadzibabic, Yefsah, Bakr, and
  Zwierlein]{Cheuk12}
Lawrence~W. Cheuk, Ariel~T. Sommer, Zoran Hadzibabic, Tarik Yefsah, Waseem~S.
  Bakr, and Martin~W. Zwierlein.
\newblock \emph{Phys. Rev. Lett.}, 109:\penalty0 095302, 2012.

\bibitem[Coutinho et~al.(1997)Coutinho, Nogami, and Perez]{Coutinho97}
F.~A.~B. Coutinho, Y.~Nogami, and J.~Fernando Perez.
\newblock \emph{J. Phys. A.: Math. Gen.}, 30:\penalty0 3937, 1997.

\bibitem[Coutinho et~al.(2004)Coutinho, Nogami, Tomio, and Toyama]{Coutinho04}
F.~A.~B. Coutinho, Y.~Nogami, Lauro Tomio, and F.~M. Toyama.
\newblock \emph{J. Phys. A.: Math. Gen.}, 37:\penalty0 10653, 2004.

\bibitem[Coutinho et~al.(2009)Coutinho, Nogami, and Toyama]{Coutinho09}
F.~A.~B. Coutinho, Y.~Nogami, and F.~M. Toyama.
\newblock \emph{Rev. Bras. Ens. de Fis.}, 31:\penalty0 4302, 2009.

\bibitem[Dalibard et~al.(2011)Dalibard, Gerbier, Juzeli\=unas, and
  \"Ohberg]{Dalibard11}
Jean Dalibard, Fabrice Gerbier, Gediminas Juzeli\=unas, and Patrik \"Ohberg.
\newblock \emph{Rev. Mod. Phys.}, 83:\penalty0 1523, 2011.

\bibitem[Dresselhaus(1955)]{Dressel55}
G.~Dresselhaus.
\newblock \emph{Phys. Rev.}, 100\penalty0 (2):\penalty0 580, 1955.

\bibitem[Galitski and Spielman(2013)]{Galitski13}
Victor Galitski and Ian~B. Spielman.
\newblock \emph{Nature}, 494:\penalty0 49, 2013.

\bibitem[Garc\'ia-Ravelo et~al.(2012)Garc\'ia-Ravelo, Schulze-Halberg, and
  Trujillo]{Ravelo12}
J.~Garc\'ia-Ravelo, A.~Schulze-Halberg, and A.~L. Trujillo.
\newblock \emph{J. Math. Phys.}, 53:\penalty0 102101, 2012.

\bibitem[Griffiths and Walborn(1999)]{Griffiths99}
David Griffiths and Stephen Walborn.
\newblock \emph{Am. J. Phys.}, 67:\penalty0 446, 1999.

\bibitem[Griffiths(1993)]{Griffiths93}
David~J. Griffiths.
\newblock \emph{J. Phys. A.: Math. Gen.}, 26:\penalty0 2265, 1993.

\bibitem[Herczy\'nski(1989)]{Herczynski89}
Jan Herczy\'nski.
\newblock \emph{J. Operator Theory}, 21:\penalty0 273, 1989.

\bibitem[Hughes(1997)]{Hughes97}
Rhonda~J. Hughes.
\newblock \emph{Rep. Math. Phys.}, 39\penalty0 (3):\penalty0 425, 1997.

\bibitem[Juzeli\=unas et~al.(2010)Juzeli\=unas, Ruseckas, and
  Dalibard]{Juzeliun10}
Gediminas Juzeli\=unas, Julius Ruseckas, and Jean Dalibard.
\newblock \emph{Phys. Rev. A}, 81:\penalty0 053403, 2010.

\bibitem[Krein(1947)]{Krein47}
M.~G. Krein.
\newblock \emph{Rec. Math. (Mat. Sbornik) N.S.}, 20(62):\penalty0 431, 1947.

\bibitem[Lin et~al.(2009)Lin, Compton, Jim\'enez-Garc\'ia, Porto, and
  Spielman]{Lin09}
Y.-J. Lin, R.~L. Compton, K.~Jim\'enez-Garc\'ia, J.~V. Porto, and I.~B.
  Spielman.
\newblock \emph{Nature}, 462:\penalty0 628, 2009.

\bibitem[Lin et~al.(2011)Lin, Jim\'enez-Garc\'ia, and Spielman]{Lin11}
Y.-J. Lin, K.~Jim\'enez-Garc\'ia, and I.~B. Spielman.
\newblock \emph{Nature}, 471:\penalty0 83, 2011.

\bibitem[Malamud and Schm\"udgen(2012)]{Malamud12}
Mark~M. Malamud and Konrad Schm\"udgen.
\newblock \emph{J. Func. Anal.}, 263:\penalty0 3144, 2012.

\bibitem[Reed and Simon(1975)]{Reed75}
M.~Reed and B.~Simon.
\newblock \emph{Methods of {M}odern {M}athematical {P}hysics {II}: {F}ourier
  {A}nalysis, {S}elf-{A}djointness}, volume~2.
\newblock Academic Press, Inc. (London) LTD., 1975.

\bibitem[Reed and Simon(1980)]{Reed80}
M.~Reed and B.~Simon.
\newblock \emph{Methods of {M}odern {M}athematical {P}hysics {I}: {F}unctional
  {A}nalysis}, volume~1.
\newblock Academic Press, Inc. (London) LTD., 1980.

\bibitem[\v{S}eba(1986)]{Seba86}
P.~\v{S}eba.
\newblock \emph{Czech. J. Phys. B}, 36:\penalty0 667, 1986.

\bibitem[Wang et~al.(2012)Wang, Yu, Fu, Miao, Huang, Chai, Zhai, and
  Zhang]{Wang12}
Pengjun Wang, Zeng-Qiang Yu, Zhengkun Fu, Jiao Miao, Lianghui Huang, Shijie
  Chai, Hui Zhai, and Jing Zhang.
\newblock \emph{Phys. Rev. Lett.}, 109:\penalty0 095301, 2012.

\end{thebibliography}

\providecommand{\noopsort}[1]{}\providecommand{\singleletter}[1]{#1}%

\end{document}